\documentclass[11pt]{article}
\pdfoutput=1

\usepackage[margin=1.0in]{geometry}

\usepackage{url}

\usepackage{amsmath,amssymb,euscript}
\usepackage{dsfont}

\usepackage{graphicx}
\usepackage{color}

\usepackage{ntheorem}

\usepackage{float}

\floatstyle{ruled}
\newfloat{structure}{Hhtb}{los}
\floatname{structure}{Structure}

\usepackage{subfig}
\usepackage{algorithmic}
\usepackage{algorithm}
\usepackage{wrapfig}

\newcounter{prob}
        {\end{list}}
\newenvironment{proof}[1][{}]{
  \begin{trivlist}\item[]\textit{Proof #1}\quad}%
  {\hfill\hspace*{\fill}~$\square$\end{trivlist}}

\theorembodyfont{\normalfont}
\newtheorem{thm}{Theorem}[section]

\newtheorem{lem}[thm]{Lemma}
\newtheorem{de}[thm]{Definition}
\newtheorem{cor}[thm]{Corollary}

\definecolor{turquoise}{cmyk}{0.65,0,0.1,0.1}

\newcommand{\rawdef}[1]{\emph{#1}} 
\newcommand{\defn}[1]{\rawdef{#1}\index{#1}}

\newcommand{\Corref}[1]{Corollary~\ref{#1}}

\newcommand{\Defref}[1]{Definition~\ref{#1}}
\newcommand{\Eqnref}[1]{Equation~\eqref{#1}}
\newcommand{\Figref}[1]{Figure~\ref{#1}}
\newcommand{\Lemref}[1]{Lemma~\ref{#1}}

\newcommand{\Secref}[1]{Section~\ref{#1}}

\newcommand{\Thmref}[1]{Theorem~\ref{#1}}

\DeclareMathOperator{\convh}{conv}
\newcommand{\convhull}[1]{\convh(#1)}
\newcommand{\pwrset}[1]{2^{#1}} 

\newcommand{\reel}{\mathbb{R}}

\newcommand{\rdee}{\reel^d}
\newcommand{\rem}{\reel^m}

\newcommand{\norm}[1]{\left\|#1\right\|}
\newcommand{\infnorm}[1]{\norm{#1}_{\infty}}
\newcommand{\abs}[1]{\left|#1\right|}
\newcommand{\transp}[1]{{#1}^\mathsf{T}}

\DeclareMathOperator{\sgn}{sign}

\newcommand{\inv}[1]{{#1}^{-1}}

\newcommand{\bdry}[1]{\partial{#1}}
\DeclareMathOperator{\bdo}{bd}
\newcommand{\bd}[1]{\bdo(#1)} 

\DeclareMathOperator{\starr}{star}
\newcommand{\str}[1]{\starr(#1)}

\newcommand{\asimplex}[1]{\{#1\}} 
\newcommand{\simplex}[1]{[#1]} 
\newcommand{\carrier}[1]{\abs{#1}} 

\newcommand{\seg}[2]{\simplex{#1,#2}} 

\newcommand{\ambdim}{N}

\newcommand{\amb}{\reel^{\ambdim}} 

\newcommand{\gdist}{d} 
\newcommand{\dist}[2]{\gdist(#1,#2)}
\newcommand{\gdistG}[1]{\gdist_{#1}} 
\newcommand{\distG}[3]{\gdist_{#1}(#2,#3)}

\newcommand{\distEd}[2]{\distG{\rdee}{#1}{#2}}
\newcommand{\gdistEm}{\gdistG{\rem}} 
\newcommand{\distEm}[2]{\distG{\rem}{#1}{#2}}

\newcommand{\close}[1]{\overline{#1}} 
 
\newcommand{\ball}[2]{B(#1;#2)} 
\newcommand{\cball}[2]{\close{B}(#1;#2)} 
\newcommand{\spaceball}[3]{B_{#1}(#2;#3)} 
\newcommand{\cspaceball}[3]{\close{B}_{#1}(#2;#3)} 

\newcommand{\ballEm}[2]{\spaceball{\rem}{#1}{#2}} 
\newcommand{\cballEm}[2]{\cspaceball{\rem}{#1}{#2}} 

\DeclareMathOperator{\aff}{aff} 
\newcommand{\affhull}[1]{\aff(#1)}

\newcommand{\angleop}[2]{\angle(#1,#2)}

\newcommand{\pts}{\mathsf{P}}
\newcommand{\tpts}{\tilde{\pts}}
\newcommand{\dipts}{\mathsf{P}_I}
\newcommand{\sdipts}{\mathsf{P}_J}
\newcommand{\qpts}{\mathsf{Q}}

\DeclareMathOperator{\vol}{vol}

\DeclareMathOperator{\Del}{Del}
\newcommand{\del}[2]{\Del_{#1}(#2)} 
\newcommand{\delPd}{\del{\gdist}{\pts}} 
\newcommand{\delof}[1]{\Del(#1)}
\newcommand{\delP}{\delof{\pts}}

\newcommand{\relDel}[2]{\Del^{#1}(#2)} 

\newcommand{\tdelta}{\tilde{\delta}} 
\newcommand{\relconst}{\xi} 
\newcommand{\trelconst}{\xi} 
\newcommand{\pertconst}{\rho}

\newcommand{\samconst}{\epsilon}
\newcommand{\tsamconst}{\tilde{\epsilon}}
\newcommand{\sparseconst}{\mu_0} 
\newcommand{\sparsity}{\lambda} 
 
\newcommand{\protconst}{\nu_0} 
\newcommand{\localconst}{\eta} 

\newcommand{\pert}{\zeta} 
\newcommand{\pertiso}{\overset{\pert}{\cong}}

\newcommand{\incl}{\iota} 
\DeclareMathOperator{\interior}{int}
\newcommand{\intr}[1]{\interior{#1}}

\newcommand{\sing}[2]{s_{#1}(#2)}
\newcommand{\pseudoinv}[1]{#1^\dagger}

\newcommand{\splxs}{\sigma}
\newcommand{\tsplxs}{\tilde{\sigma}}
\newcommand{\splxt}{\tau}
\newcommand{\tsplxt}{\tilde{\tau}}

\newcommand{\splx}[1]{\sigma^{#1}} 
\newcommand{\tsplx}[1]{\tilde{\sigma}^{#1}} 

\newcommand{\splxjoin}[2]{{#1}*{#2}}
\newcommand{\normhull}[1]{N(#1)}

\newcommand{\opface}[2]{#2_{#1}} 
\newcommand{\splxsp}{\opface{p}{\splxs}}
\newcommand{\thickbnd}{\Upsilon_0}
\newcommand{\thickness}[1]{\Upsilon(#1)}

\newcommand{\splxalt}[2]{D(#1,#2)} 
\newcommand{\longedge}[1]{\Delta(#1)}
\newcommand{\shortedge}[1]{L(#1)}
\newcommand{\circrad}[1]{R(#1)}
\newcommand{\circcentre}[1]{C(#1)}

\newcommand{\X}{X} 

\usepackage{enumitem}
\usepackage[pagebackref=true]{hyperref}

\title{The stability of Delaunay triangulations}
\author{Jean-Daniel Boissonnat \and Ramsay Dyer \and Arijit Ghosh}

\begin{document}


\maketitle


\begin{abstract}
  We introduce a parametrized notion of genericity for Delaunay
  triangulations which, in particular, implies that the Delaunay
  simplices of $\delta$-generic point sets are thick.  Equipped with
  this notion, we study the stability of Delaunay triangulations under
  perturbations of the metric and of the vertex positions. We quantify
  the magnitude of the perturbations under which the Delaunay
  triangulation remains unchanged.

\end{abstract}


\pagenumbering{roman}

\tableofcontents
\clearpage

\pagenumbering{arabic}

%

\section{Introduction}

One of the central properties of Delaunay complexes, which was
demonstrated when they were introduced \cite{delaunay1934}, is that
under a very mild assumption they are embedded, i.e., they define a
triangulation of Euclidean space.  The required assumption is that
there are not too many cospherical points; the points are
``generic''. The assumption is not considered limiting because, as
Delaunay showed, an arbitrarily small affine perturbation can
transform any given point set into one that is generic.

Given the assumption of a generic point set, we are assured that the
Delaunay complex defines a triangulation, but a couple of issues arise
when working with these triangulations. One is that the Delaunay
triangulation can be highly sensitive to the exact location of the
points. For example, the Delaunay triangulation of a point set might
be different if a coordinate transform is first performed using
floating point arithmetic.

Another problem concerns the geometric quality of the simplices in the
triangulation. We define the \defn{thickness} of a simplex as a number
proportional to the ratio of the smallest altitude to the longest edge
length of the simplex, and we demonstrate why this is a useful measure
of the geometric quality of the simplex. For points in the plane, if
there is an upper bound on the ratio of the radius of a Delaunay ball
to the length of the shortest edge of the corresponding triangle, then
there is a lower bound on the thickness of any Delaunay
triangle. However, when there are three or more spatial dimensions,
the thickness of Delaunay simplices may become arbitrarily small in
spite of any bound on the circumradius to shortest edge length.

Both of these issues are shown to be related to points being close to
a degenerate (non-generic) configuration. We parameterize Delaunay's
original definition of genericity, saying that a point set $\pts
\subset \rem$ is $\delta$-generic if every $m$-simplex in the Delaunay
complex has a Delaunay ball that is at a distance greater than
$\delta$ to the remaining points in $\pts$. We show that a bound on
$\delta$ leads to a bound on the thickness of the Delaunay simplices,
and also that the Delaunay complex itself is stable with respect to
perturbations of the points or of the metric, provided the
perturbation is small enough with respect to $\delta$ in a way that we
quantify. In a companion paper \cite{boissonnat2013flatpert.inria}, we
develop a perturbation algorithm to produce $\delta$-generic point
sets.

The stability of Delaunay triangulations has not previously been
studied in this way. Related work can be found in the context of
kinetic data structures \cite{agarwal2010scg} or in the context of
robust computation \cite{bandyopadhyay2004}, and in particular, the
concept of protection we introduce in \Secref{sec:param.gen} is
embodied in the guarded insphere predicate which has been employed in
a controlled perturbation algorithm for 2D Delaunay
triangulation~\cite{funke2005}. 
%

Our interest in the problem of near-degeneracy in Delaunay complexes
stems from work on triangulating Riemannian manifolds. An established
technique is to compute the triangulation locally at each point in an
approximating Euclidean metric, and then perform manipulations to
ensure that the local triangulations fit together consistently
\cite{boissonnat2011aniso.tr,boissonnat2011tancplx}. The reason the
manipulations are necessary is exactly the problem of the instability
of the Delaunay triangulation, and sometimes this is most conveniently
described as an instability with respect to a perturbation of the
local Euclidean metric.


Although we make no explicit reference to Voronoi diagrams, the
Delaunay complexes we study can be equivalently defined as the nerve
of the Voronoi diagram associated with the metric under consideration.
We provide criteria for ensuring that the Delaunay complex is a
triangulation without explicit requirements on the properties of the
Voronoi diagram \cite{edelsbrunner1997rdt}, in contrast to a common
practice
in related
work~\cite{leibon2000,labelle2003,cheng2005,dyer2008sgp,canas2012}.

After presenting background material in \Secref{sec:background}, we
introduce the concept of $\delta$-generic point sets for Euclidean
Delaunay triangulations in \Secref{sec:param.gen}. We show that
Delaunay simplices of $\delta$-generic point sets are thick; they
satisfy a quality bound. Then in \Secref{sec:stability} we quantify
how $\delta$-genericity leads to robustness in the Delaunay
triangulation when either the points or the metric are perturbed. The
primary challenge is bounding the displacement of simplex
circumcentres. We conclude with some remarks on the construction and
application of $\delta$-generic point sets.

%

\section{Background}
\label{sec:background}

Within the context of the standard $m$-dimensional Euclidean space
$\rem$, when distances are determined by the standard norm,
$\norm{\cdot}$, we use the following conventions.  The distance
between a point $p$ and a set $\X \subset \rem$, is the infimum of the
distances between $p$ and the points of $\X$, and is denoted
$\distEm{p}{\X}$.  We refer to the distance between two points $a$ and
$b$ as $\norm{b-a}$ or $\distEm{a}{b}$ as convenient. A ball
$\ballEm{c}{r} = \{ x \, | \, \norm{x-c}< r \}$ is \textbf{open}, and
$\cballEm{c}{r}$ is its topological closure. 
We will consider other metrics besides the Euclidean one. A generic
metric is denoted $\gdist$, and the associated open and closed balls
are $\ball{c}{r}$, and $\cball{c}{r}$. 
Generally, we denote the
topological closure of a set $\X$ by $\close{\X}$, the interior by
$\intr{\X}$, and the boundary by $\bdry{\X}$. The convex hull is
denoted $\convhull{\X}$, and the affine hull is $\affhull{\X}$.

If $U$ and $V$ are vector subspaces of $\rem$, with $\dim U \leq \dim
V$, the \defn{angle} between them is defined by 
\begin{equation}
  \label{eq:angle.def}
  \sin \angleop{U}{V} = \sup_{\substack{u \in U\\ \norm{u}=1}} \norm{ u - \pi_V u},
\end{equation}
where $\pi_V$ is the orthogonal projection onto $V$.
This is the largest principal angle between $U$ and $V$. The angle
between affine subspaces $K$ and $H$ is defined as the angle between
the corresponding parallel vector subspaces. 

\subsection{Sampling parameters and perturbations}

The structures of interest will be built from a finite set $\pts
\subset \rem$, which we consider to be a set of \defn{sample points}.
If $D \subset \rem$ is a bounded set, then $\pts$
is an \defn{$\samconst$-sample set} for $D$ if $\distEm{x}{\pts} <
\samconst$ for all $x \in \close{D}$. We say that $\samconst$ is a
\defn{sampling radius} for $D$ satisfied by $\pts$.  If no domain $D$
is specified, we say $\pts$ is an $\samconst$-sample set if
$\distEm{x}{\pts \cup \bdry{\convhull{\pts}}} < \samconst$ for all $x
\in \convhull{\pts}$. 
Equivalently, $\pts$ is an $\samconst$-sample
set if it satisfies a sampling radius $\samconst$ for
\begin{equation*}
  D_\samconst(\pts) = \{ x \in \convhull{\pts} \, | \,
  \distEm{x}{\bdry{\convhull{\pts}}} \geq \samconst \}.
\end{equation*}
The set $\pts$ is \defn{$\sparsity$-separated} if $\distEm{p}{q} >
\sparsity$ for all $p,q \in \pts$. We usually assume that the sparsity
of a $\samconst$-sample set is proportional to $\samconst$, thus:
$\sparsity = \sparseconst \samconst$.

We consider a perturbation of the points $\pts \subset \rem$ given by
a function $\pert: \pts \to \rem$. If $\pert$ is such that
$\distEm{p}{\pert(p)} \leq \pertconst$, we say that $\pert$ is a
\defn{$\pertconst$-perturbation}. As a notational convenience, we
frequently define $\tpts = \pert(\pts)$, and let $\tilde{p}$ represent
$\pert(p) \in \tpts$. We will only be considering
$\pertconst$-perturbations where $\pertconst$ is less than half the
sparsity of $\pts$, so $\pert: \pts \to \tpts$ is a bijection.

Points in $\pts$ which are not on the boundary of $\convhull{\pts}$
are \defn{interior points} of $\pts$.


\subsection{Simplices}

Given a set of $j+1$ points $\asimplex{p_0, \ldots, p_j} \subset \pts
\subset \rem$, a (geometric) \defn{$j$-simplex} $\splxs =
\simplex{p_0, \ldots, p_j}$ is defined by the convex hull: $\splxs =
\convhull{\asimplex{p_0, \ldots, p_j}}$. The points $p_i$ are the
\defn{vertices} of $\splxs$. Any subset $\asimplex{p_{i_0}, \ldots,
  p_{i_k}}$ of $\asimplex{p_0, \ldots, p_j}$ defines a $k$-simplex
$\splxt$ which we call a \defn{face} of $\splxs$. We write $\splxt
\leq \splxs$ if $\splxt$ is a face of $\splxs$, and $\splxt < \splxs$
if $\splxt$ is a \defn{proper face} of $\splxs$, i.e., if the vertices
of $\splxt$ are a proper subset of the vertices of $\splxs$.

The \defn{boundary} of $\splxs$, is the union of its proper faces:
$\bdry{\splxs} = \bigcup_{\splxt < \splxs}\splxt$. In general this is
distinct from the topological boundary defined above, but we denote it
with the same symbol. The \defn{interior} of $\splxs$ is
$\intr{\splxs} = \splxs \setminus \bdry{\splxs}$. Again this is
generally different from the topological interior.  In particular, a
$0$-simplex $p$ is equal to its interior: it has no boundary.  Other
geometric properties of $\splxs$ include its diameter (its longest
edge), $\longedge{\splxs}$, and its shortest edge,
$\shortedge{\splxs}$.

For any vertex $p \in \splxs$, the \defn{face oppposite} $p$ is the
face determined by the other vertices of $\splxs$, and is denoted
$\splxsp$. If $\splxt$ is a $j$-simplex, and $p$ is not a vertex of
$\splxt$, we may construct a $(j+1)$-simplex $\splxs =
\splxjoin{p}{\splxt}$, called the \defn{join} of $p$ and $\splxt$. It
is the simplex defined by $p$ and the vertices of $\splxt$, i.e.,
$\splxt = \splxsp$.

Our definition of a simplex has made an important departure from
standard convention: we do not demand that the vertices of a simplex
be affinely independent. A $j$-simplex $\splxs$ is a \defn{degenerate
  simplex} if $\dim \affhull{\splxs} < j$. If we wish to emphasise
that a simplex is a $j$-simplex, we write $j$ as a superscript:
$\splx{j}$; but this always refers to the \defn{combinatorial}
dimension of the simplex.


A \defn{circumscribing ball} for a simplex $\splxs$ is any
$m$-dimensional ball that contains the vertices of $\splxs$ on its
boundary. If $\splxs$ admits a circumscribing ball, then it has a
\defn{circumcentre}, $\circcentre{\splxs}$, which is the centre of the
smallest circumscribing ball for $\splxs$. The radius of this ball is
the \defn{circumradius} of $\splxs$, denoted $\circrad{\splxs}$.  In
general a degenerate simplex may not have a circumcentre and
circumradius, but in the context of the Euclidean Delaunay complexes
we will work with, the degenerate simplices we may encounter do have
these properties.
We will make use of the affine space $\normhull{\splxs}$ composed of
the centres of the balls that circumscribe $\splxs$. This space is
orthogonal to $\affhull{\splxs}$ and intersects it at the circumcentre
of $\splxs$. Its dimension is $m - \dim \affhull{\splxs}$.

The \defn{altitude} of a vertex $p$ in $\splxs$ is
$\splxalt{p}{\splxs} = \distEm{p}{\affhull{\splxsp}}$. A poorly-shaped
simplex can be characterized by the existence of a relatively small
altitude. The \defn{thickness} of a $j$-simplex $\splxs$ is the
dimensionless quantity
\begin{equation*}
  \thickness{\splxs} =
  \begin{cases}
    1& \text{if $j=0$} \\
    \min_{p \in \splxs} \frac{\splxalt{p}{\splxs}}{j
      \longedge{\splxs}}& \text{otherwise.}
  \end{cases}
\end{equation*}
We say that $\splxs$ is $\thickbnd$-thick, if $\thickness{\splxs} \geq
\thickbnd$. If $\splxs$ is $\thickbnd$-thick, then so are
all of its faces. Indeed if $\splxt \leq \splxs$, then the smallest
altitude in $\splxt$ cannot be smaller than that of $\splxs$, and also
$\longedge{\splxt} \leq \longedge{\splxs}$. 

Our definition of thickness is essentially the same as that employed
by Munkres~\cite{munkres1968}. Munkres defined the thickness of
$\splxs^j$ as $\frac{r(\splxs^j)}{\longedge{\splxs^j}}$, where
$r(\splxs^j)$ is the radius of the largest contained ball centred at
the barycentre. This definition of thickness turns out to be equal to
$\frac{j}{j+1}\thickness{\splxs^j}$.

Whitney~\cite{whitney1957} employed a volume-based measure of simplex
quality, and variations on this, typically referred to as
\defn{fatness}, have been popular in works on higher dimensional
Delaunay-based meshing
\cite{cheng2000,li2003,boissonnat2011aniso.tr}. 
We find a direct bound on the altitudes to be more convenient, because
it yields a cleaner and tighter connection between the geometry and
the linear algebra of simplices. Typically, a bound on some geometric
displacement related to a simplex is obtained by bounding the inverse
of a matrix associated with the simplex, and thickness is well suited
for this task.

As a motivating example, consider the problem of bounding the angle
between the affine hull of a simplex and an affine space that lies
close to all the vertices of the simplex. Such a bound is relevant
when meshing submanifolds of Euclidean space, for example, where it is
desired that the affine hulls of the simplices are in agreement with
the nearby tangent spaces of the manifold.

Whitney~\cite[p. 127]{whitney1957} obtained such a bound, which
manifestly depends on the quality of the simplex. Using thickness as a
quality measure we obtain a sharper result: 
\begin{lem}[Whitney angle bound]
  \label{lem:whitney.approx}
  Suppose $\splxs$ is a $j$-simplex whose vertices all lie within a
  distance $\localconst$ from a $k$-dimensional affine space, $H
  \subset \rem$, with $k \geq j$. Then
  \begin{equation*}
    \sin \angleop{\affhull{\splxs}}{H} \leq
    \frac{2\localconst}{\thickness{\splxs}\longedge{\splxs}}.
  \end{equation*}
\end{lem}
The idea of the proof is to express the unit vector $u$ in
\Eqnref{eq:angle.def} in terms of a basis for $\affhull{\splxs}$ given
by the edges of $\splxs$ that emenate from some arbitrarily chosen
vertex. The projection $\tilde{u}$ of $u$ into $H$ can then be
expressed in terms of the projected basis vectors, using the same
vector of coefficients. Since the vertices of $\splxs$ all lie close
to $H$, the projected basis vectors do not differ significantly from
the originals, so bounding the magnitude of the difference between $u$
and $\tilde{u}$ comes down to bounding the magnitude of the vector of
coefficients of the unit vector $u$.
This bound depends on how well-conditioned the basis is, and this
is closely related to the thickness of $\splxs$.

These observations can be conveniently expressed and made concrete
in terms of the singular values of a matrix. An excellent introduction
to singular values can be found in the book by Trefethen and Bau
\cite[Ch. 4 \& 5]{trefethen1997}, but for our purposes we are
primarily concerned with the largest and the smallest singular values,
which we now describe.

We denote the $i^{th}$ singular value of a matrix $A$ by
$\sing{i}{A}$. The singular values are non-negative and ordered by
magnitude. The largest singular value can be defined as $\sing{1}{A}
=\sup_{\norm{x}=1} \norm{Ax}$; it is the magnitude of the largest
vector in the range of the unit sphere. The first singular value also
defines the operator norm: $\norm{A} = \sing{1}{A}$.  The standard
observation that a bound on the norms of the columns of $A$ yields a
bound on $\norm{A}$ is obtained by a short calculation.
\begin{lem}
  \label{lem:col.mat.norm}
  If $\localconst > 0$ is the least upper bound on the norms of the columns
  of an $m \times j$ matrix $A$, then $\localconst \leq \norm{A} \leq
  \sqrt{j}\localconst$.
\end{lem}
We will also be interested in obtaining a lower bound on the smallest
singular value which, for an $m \times j$ matrix $A$ with $j \leq m$,
may be defined as $\sing{j}{A} = \inf_{\norm{x} = 1} \norm{Ax}$.

From the given definitions, one can verify that if $A$ is an
invertible $m \times m$ matrix, then $\sing{1}{A^{-1}} =
\sing{m}{A}^{-1}$, but it is convenient to also accommodate non-square
matrices, corresponding to simplices that are not full dimensional.
If $A$ is an $m \times j$ matrix of rank $j \leq m$, then the
\defn{pseudo-inverse} $\pseudoinv{A} = (\transp{A}A)^{-1}\transp{A}$
is the unique left inverse of $A$ whose kernel is the orthogonal
complement of the column space of $A$. 
We have the following general observation:
\begin{lem}
  \label{lem:svd.pseudo.inv}
  If $A$ is an $m \times j$ matrix of rank $j \leq m$, then
$
    s_i(\pseudoinv{A}) = s_{j-i+1}(A)^{-1}.
$
\end{lem}

The columns of $A$ form a basis for the column space of $A$.  The
pseudo-inverse can also be described in terms of the \defn{dual
  basis}. If we denote the columns of $A$ by $\{a_i\}$, then the
$i^{\text{th}}$ dual vector, $w_i$, is the unique vector in the column
space of $A$ such that $\transp{w}_i a_i = 1$ and $\transp{w}_i a_j =
0$ if $i \neq j$. Then $\pseudoinv{A}$ is the $j \times m$ matrix
whose $i^{\text{th}}$ row is $\transp{w}_i$.

\begin{figure}[ht]
  \begin{center}
    \includegraphics[width=.4\columnwidth]{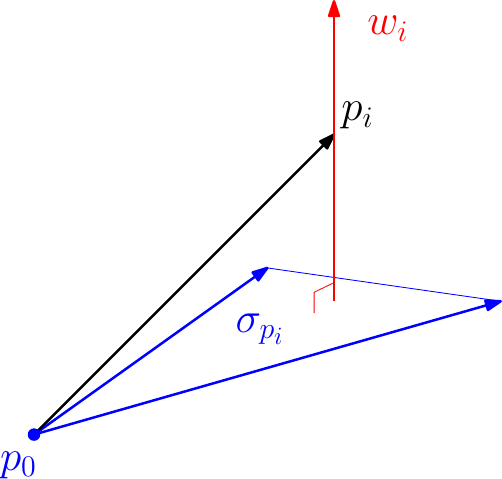} 
  \end{center}
  \caption{Choosing $p_0$ as the origin, the edges emenating from
    $p_0$ in $\splxs = \simplex{p_0, \ldots, p_j}$ form a basis for
    $\affhull{\splxs}$. The proof of \Lemref{lem:bound.skP}
    demonstrates that the dual basis $\{w_i\}$ consists of vectors
    that are orthogonal to the facets, and with magnitude equal to the
    inverse of the corresponding altitude.}
  \label{fig:vect.altitude}
\end{figure}
By exploiting a close connection between the altitudes of a simplex
and the vectors dual to a basis defined by the simplex, we obtain the
following key lemma that relates the thickness of a simplex to the
smallest singular value of an associated matrix:
\begin{lem}[Thickness and singular value]
  \label{lem:bound.skP}
  Let $\splxs = \simplex{p_0, \ldots, p_j}$ be a non-de\-generate
  $j$-simplex in $\rem$, with $j>0$, and let $P$ be the $m \times j$
  matrix whose $i^\text{th}$ column is $p_i - p_0$. Then the
  $i^{\text{th}}$ row of $\pseudoinv{P}$ is given by $\transp{w}_i$,
  where $w_i$ is orthogonal to $\affhull{\opface{p_i}{\splxs}}$, and
  \begin{equation*}
    \norm{w_i} = \splxalt{p_i}{\splxs}^{-1}.
  \end{equation*}
  We have the following bound on the smallest singular value of $P$:
  \begin{equation*}
    \sing{j}{P} \geq \sqrt{j} \thickness{\splxs}\longedge{\splxs}.
  \end{equation*}
\end{lem}
\begin{proof}
  By the definition of $\pseudoinv{P}$, it follows that $w_i$ belongs
  to the column space of $P$, and it is orthogonal to all
  $(p_{i'}-p_0)$ for $i' \neq i$. Let $u_i = w_i/\norm{w_i}$. By the
  definition of $w_i$, we have $\transp{w}_i(p_i-p_0) = 1 =
  \norm{w_i}\transp{u}_i(p_i-p_0)$. By the definition of the altitude
  of a vertex, we have $\transp{u}_i(p_i-p_0) =
  \splxalt{p_i}{\splxs}$. Thus $\norm{w_i} =
  \splxalt{p_i}{\splxs}^{-1}$.
  Since
  \begin{equation*}
    \max_{1 \leq i \leq j}\splxalt{p_i}{\splxs}^{-1} = \left( \min_{1
      \leq i \leq j}\splxalt{p_i}{\splxs} \right)^{-1} = (j
    \thickness{\splxs} \longedge{\splxs})^{-1},
  \end{equation*}
  \Lemref{lem:col.mat.norm}, yields  
  \begin{equation*}
    \sing{1}{\pseudoinv{P}} \leq (\sqrt{j} \thickness{\splxs}
    \longedge{\splxs})^{-1},
  \end{equation*}
  because $\sing{i}{\transp{A}} =
  \sing{i}{A}$ for any matrix $A$.
  The stated bound on $\sing{j}{P}$ follows from
  \Lemref{lem:svd.pseudo.inv}. 
\end{proof}

The proof of \Lemref{lem:bound.skP} shows that the pseudoinverse of
$P$ has a natural geometric interpretation in terms of the altitudes
of $\splxs$, and thus the altitudes provide a convenient lower bound
on $\sing{j}{P}$. By \Lemref{lem:col.mat.norm}, $\sing{1}{P} \leq
\sqrt{j}\longedge{\splxs}$, and thus
$
  \thickness{\splxs} \leq \frac{\sing{j}{P}}{\sing{1}{P}}.
$
In other words, $\thickness{\splxs}^{-1}$ provides a convenient upper
bound on the \defn{condition number} of $P$.  Roughly speaking,
thickness imparts a kind of stability on the geometric properties of a
simplex.  This is exactly what is required when we want to show that a
small change in a simplex will not yield a large change in some
geometric quantity of interest. For example, we will use
\Lemref{lem:bound.skP} in the demonstration of
\Lemref{lem:power.close.centres}, which is the technical lemma related
to the stability of the space of circumcentres of a simplex.
\Lemref{lem:bound.skP} also facilitates a concise demonstration of
Whitney's angle bound:
\begin{proof}[of \Lemref{lem:whitney.approx}]
  Suppose $\splxs = \simplex{p_0, \ldots, p_j}$.  Choose $p_0$ as the
  origin of $\rem$, and let $U \subset \rem$ be the vector subspace
  defined by $\affhull{\splxs}$.  Let $W$ be the $k$-dimensional
  subspace parallel to $H$, and let $\pi: \rem \to W$ be the
  orthogonal projection onto $W$.

  We desire an upper bound on $\norm{u - \pi u}$ for all unit vectors
  $u \in U$.
  Since the vectors $v_i =
  (p_i - p_0)$, $i \in \{1,\ldots,j\}$ form a basis for $U$, we may
  write $u = Pa$, where $P$ is the $m \times j$ matrix whose
  $i^\text{th}$ column is $v_i$, and $a \in \reel^j$ is the vector of
  coefficients. Then, defining $X =P - \pi P$, we get
  \begin{equation*}
    \norm{u - \pi u} = \norm{Xa} \leq \norm{X}\norm{a}.
  \end{equation*}

  $W$ is at a distance less than $\localconst$ from $H$, because $p_0
  \in W$ and $\distEm{p_i}{H} \leq \localconst$ for all $0 \leq i \leq
  j$. It follows that $\norm{v_i - \pi v_i} \leq 2\localconst$, and
  \Lemref{lem:col.mat.norm} yields
  \begin{equation*}
    \norm{X} \leq  2 \sqrt{j} \localconst.
  \end{equation*}
 Observing that $1 = \norm{u} = \norm{Pa} \geq s_j(P)\norm{a}$, we find
  \begin{equation*}
    \norm{a} \leq \frac{1}{s_j(P)},
  \end{equation*}
  and the result follows from \Lemref{lem:bound.skP}.
\end{proof}

\subsection{Complexes}

Given a finite set $\pts$, an \defn{abstract simplicial complex} is a
set of subsets $K \subset \pwrset{\pts}$ such that if $\splxs \in K$,
then every subset of $\splxs$ is also in $K$. 
The Delaunay complexes we study are abstract simplicial complexes, but
their simplices carry a canonical geometry induced from the inclusion
map $\incl: \pts \hookrightarrow \rem$. (We assume $\incl$ is
injective on $\pts$, and so do not distinguish between $\pts$ and
$\incl(\pts)$.) For each abstract simplex $\splxs \in K$, we have an
associated geometric simplex $\convhull{\incl(\splxs)}$, and normally
when we write $\splxs \in K$, we are referring to this geometric
object.  Occasionally, when it is convenient to emphasise a
distinction, we will write $\incl(\splxs)$ instead of $\splxs$.

Thus we view such a $K$ as a set of simplices in $\rem$, and we refer
to it as a \defn{complex}, but it is not generally a (geometric)
simplicial complex.  A geometric \defn{simplicial complex} is a finite
collection $G$ of simplices in $\amb$ such that if $\splxs \in G$,
then all of the faces of $\splxs$ also belong to $G$, and if $\splxs,
\tsplxs \in G$ and $\splxs \cap \tsplxs \neq \emptyset$, then $\splxt
= \splxs \cap \tsplxs$ is a simplex and $\splxt \leq \splxs$ and
$\splxt \leq \tsplxs$. Observe that the simplices in a geometric
simplicial complex are necessarily non-degenerate.
An abstract simplicial complex is defined from a geometric simplicial
complex in an obvious way.  A \defn{geometric realization} of an
abstract simplicial complex $K$ is a geometric simplicial complex
whose associated abstract simplicial complex may be identified with
$K$.  A geometric realization always exists for any complex. Details
can be found in algebraic topology textbooks; the book by
Munkres~\cite{munkres1984} for example.

The \defn{dimension of a complex} $K$ is the largest dimension of the
simplices in $K$. We say that $K$ is an $m$-complex, to mean that it
is of dimension $m$. The complex $K$ is a \defn{pure $m$-complex} if
it is an $m$-complex, and every simplex in $K$ is the face of an
$m$-simplex. 

The \defn{carrier} of an abstract complex $K$ is the underlying
topological space $\carrier{K}$, associated with a geometric
realization of $K$.  Thus if $G$ is a geometric realization of $K$,
then $\carrier{K} = \bigcup_{\splxs \in G} \splxs$.  For our
complexes, the inclusion map $\incl$ induces a continous map $\incl:
\carrier{K} \to \rem$, defined by barycentric interpolation on each
simplex. If this map is injective, we say that $K$ is
\defn{embedded}. In this case $\incl$ also defines a geometric
realization of $K$, and we may identify the carrier of $K$ with the
image of $\incl$.
%

A subset $K' \subset K$ is a \defn{subcomplex} of $K$ if it is also a
complex.  The \defn{star} of a subcomplex $K' \subseteq K$ is the
subcomplex generated by the simplices incident to $K'$. I.e., it is
all the simplices that share a face with a simplex of $K'$, plus all
the faces of such simplices. This is a departure from a common usage
of this same term in the topology literature. The star of $K'$ is
denoted $\str{K'}$ when there is no risk of ambiguity, otherwise we
also specify the parent complex, as in $\str{K';K}$. A simple example
of the star of a complex is depicted in \Figref{fig:star.cplx}.
\begin{figure}[ht]
  \begin{center}
    \includegraphics[width=.6\columnwidth]{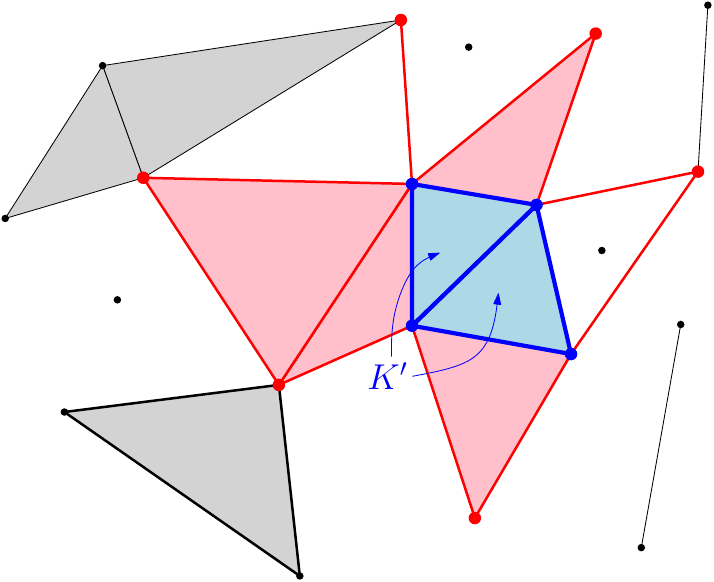} 
  \end{center}
  \caption{The star of a subcomplex $K' \subset K$ is the subcomplex
    $\str{K'} \subset K$ that consists all the simplices that share a
    face with $K'$ (this includes all of $K'$ itself), and all the
    faces of these simplices. Here we show an embedded $2$-complex,
    with all $2$-simplices shaded. The subcomplex $K'$ consists of the
    two indicated triangles, and their faces (blue). The simplices of
    $\str{K'}$ are shown in bold (red and blue). The other simplices
    do not belong to $\str{K'}$ (black).}
  \label{fig:star.cplx}
\end{figure}

A \defn{triangulation} of $\pts \subset \rem$ is an embedded complex
$K$ with vertices $\pts$ such that $\carrier{K} = \convhull{\pts}$. A
complex $K$ is a \defn{$j$-manifold complex} if the star of every
vertex is isomorphic to the star of a triangulation of $\reel^j$. In
order to exploit the local nature of the definition of a manifold
complex, it is convenient to have a local notion of triangulation for
the star of a vertex in $K$, even if the whole of $K$ is not a
triangulation of its vertices:
\begin{de}[Triangulation at a point]
  \label{def:local.triangulation}
  A complex $K$ is a \defn{triangulation at $p \in \rem$} if:
  \begin{enumerate}[noitemsep,topsep=0pt,parsep=0pt,partopsep=0pt]
  \item \label{tri:is.vtx}$p$ is a vertex of $K$.
  \item \label{tri:str.embedded}$\str{p}$ is embedded.
  \item \label{tri:interior.pt}$p$  lies in $\intr{\carrier{\str{p}}}$.
  \item \label{tri:no.conflicts}For all $\splxt \in K$, and $\splxs
    \in \str{p}$, if $(\intr{\splxt}) \cap \splxs \neq \emptyset$,
    then $\splxt \in \str{p}$.
  \end{enumerate}
\end{de}
In a general complex Condition~\ref{tri:no.conflicts} above is not a
local property, however in the case of Delaunay complexes that
intersts us here, local conditions are sufficient to verify the
condition, as we will show in \Secref{sec:local.del.tri}. Observe also
that Condition~\ref{tri:no.conflicts} also precludes intersections
with degenerate simplices, since such a simplex would have a face that
violates the conditon.

If $\splxs$ is a simplex with vertices in $\pts$, then any map $\pert:
\pts \to \tpts \subset \rem$ defines a simplex $\pert(\splxs)$ whose
verticies in $\tpts$ are the images of vertices of $\splxs$. If $K$ is
a complex on $\pts$, and $\tilde{K}$ is a complex on $\tpts$, then
$\pert$ induces a \defn{simplicial map} $K \to \tilde{K}$ if
$\pert(\splxs) \in \tilde{K}$ for every $\splxs \in K$. We denote this
map by the same symbol, $\pert$. We are interested in the 
case when $\pert$ is an \defn{isomorphism}, which means it establishes
a bijection between $K$ and $\tilde{K}$. We then say that $K$ and
$\tilde{K}$ are \defn{isomorphic}, and write $K \cong \tilde{K}$, or
$K \pertiso \tilde{K}$ if we wish to emphasise that the correspondence
is given by $\pert$.

A simplicial map $\pert: K \to \tilde{K}$ defines a continuous map
$\pert: \carrier{K} \to \carrier{\tilde{K}}$, by barycentric
interpolation on each simplex $\splxs \in K$.  We observe the
following consequence of Brouwer's invariance of domain:
\begin{lem}
  \label{lem:inject.triang}
  Suppose $K$ is a complex with vertices $\pts \subset \rem$, and
  $\tilde{K}$ a complex with vertices $\tpts \subset \rem$. Suppose
  also that $K$ is a triangulation at $p \in \pts$, and that $\pert:
  \pts \to \tpts$ induces an injective simplicial map $\str{p} \to
  \str{\pert(p)}$.
  If $\str{\pert(p)}$ is embedded, then 
  \begin{equation*}
    \pert(\str{p}) = \str{\pert(p)},
  \end{equation*}
  and $\pert(p)$ is an interior point of $\tpts$.
\end{lem}
\begin{proof}
  We need to show that $\str{\pert(p)} \subseteq \pert(\str{p})$.
  Since $\str{p}$ is embedded, $\pert$ defines a
  continuous map $\pert: \carrier{\str{p}} \to
  \carrier{\str{\pert(p)}}$ that is injective on each simplex. Since
  $\str{\pert(p)}$ is also embedded, this continuous map is injective
  on $\carrier{\str{p}}$.  Since $K$ is a triangulation at $p$, there
  is an open ball $B$ centred at $p$ such that $B \subset
  \intr{\carrier{\str{p}}}$. Then $\pert|_B : \rem \supset B \to
  \pert(B) \subset \rem$ is a homeomorphism by Brouwer's invariance of
  domain~\cite[Ch. XVII]{dugundji1966}. It follows that $\pert(p)$ is
  an interior point of $\tpts$.

  Suppose $\splxs \in \str{\pert(p)}$ and $\pert(p)$ is a vertex of
  $\splxs$. Then, since $\splxs$ is not degenerate, there is a point
  $x \in \pert(B) \cap \intr{\splxs}$, and from the above argument,
  $x$ also lies in the interior of some simplex $\tsplxt \in
  \pert(\str{p}) \subseteq \str{\pert(p)}$. Since $\str{\pert(p)}$ is
  embedded, $\tsplxt \cap \splxs$ is a face of $\splxs$ and of
  $\tsplxt$, but since $x$ is in the interior of both simplices, it
  must be that $\tsplxt = \splxs$. Thus $\splxs \in \pert(\str{p})$.

  If $\splxs \in \str{\pert(p)}$, then there is some $\splxt \in
  \str{\pert(p)}$ such that $\pert(p)$ is a vertex of $\splxt$ and
  $\splxs \leq \splxt$. Since $\splxt \in \pert(\str{p})$, we also
  have $\splxs \in \pert(\str{p})$, by the definition of a simplicial
  map. 
\end{proof}

%
\section{Parameterized genericity}
\label{sec:param.gen}

In this section we examine the Delaunay complex of $\pts \subset
\rem$, taking the view that poorly-shaped simplices arise from almost
degenerate configurations of points. We introduce the concept of a
protected Delaunay ball, which leads to a parameterized definition
of genericity. We then show that a lower bound on the protection of
the maximal simplices yields a lower bound on their thickness.

\subsection{The Delaunay complex}

An \defn{empty ball} is one that contains no point from $\pts$. 
\begin{de}[Delaunay complex]
  \label{def:Delaunay.complex}
  A \defn{Delaunay ball} is a maximal empty ball. Specifically, $B =
  \ballEm{x}{r}$ is a Delaunay ball if any empty ball centred at $x$
  is contained in $B$. A simplex $\splxs$ is a \defn{Delaunay
    simplex} if there exists some Delaunay ball $B$ such that the
  vertices of $\splxs$ belong to $\bdry{B}\cap \pts$.  The
  \defn{Delaunay complex} is the set of Delaunay simplices, and is
  denoted $\delP$.
\end{de}
The Delaunay complex has the combinatorial structure of an abstract
simplicial complex, but $\delP$ is embedded only when
$\pts$ satisfies appropriate genericity requirements, as discussed in
\Secref{sec:protection}. Otherwise, $\delP$ contains degenerate
simplices.  We make here some observations that are not dependent on
assumptions of genericity.

\begin{figure}[ht]
  \begin{center}
    \includegraphics[width=.4\columnwidth]{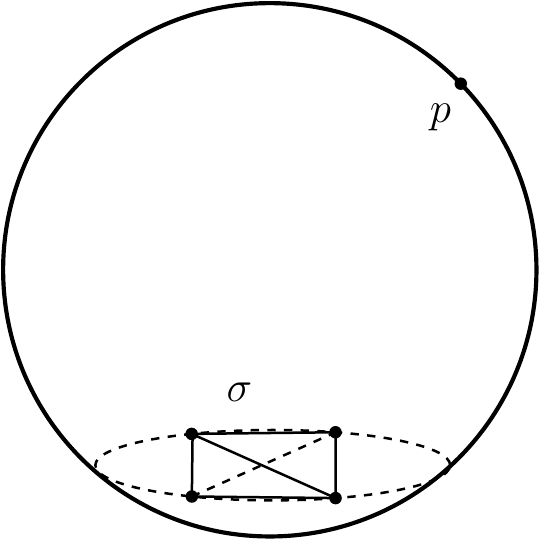} 
  \end{center}
  \caption{\Lemref{lem:maximal.splx}: If the affine hull of $\splxs$
    is not full dimensional, then a Delaunay ball has freedom to
    expand, and $\splxs$ must be the face of a higher dimensional
    Delaunay simplex.}
  \label{fig:polar.sliver}
\end{figure}
%
The union of the Delaunay simplices is $\convhull{\pts}$.  A simplex
$\splxs \in \delP$ is a \defn{boundary simplex} if all its vertices
lie on $\bdry{\convhull{P}}$. We observe
\begin{lem}[Maximal simplices]
  \label{lem:maximal.splx}
  If $\affhull{\pts} = \rem$, then every Delaunay $j$-simplex,
  $\splxs$, is a face of a Delaunay simplex $\splxs'$ with $\dim
  \affhull{\splxs'} = m$. In particular, if $j \leq m$, then $\splxs$
  is a face of a Delaunay $m$-simplex.
  If $\splxs$ is not a boundary simplex, and $\dim\affhull{\splxs} <
  m$, then there are at least two Delaunay $(j+1)$-simplices that have
  $\splxs$ as a face.
\end{lem}
\begin{proof}
  Suppose $\dim \affhull{\splxs} < m$. Let $B = \ballEm{c}{r}$ be a
  Delaunay ball for $\splxs$. Let $\ell$ be the line through $c$ and
  $\circcentre{\splxs}$. If $c = \circcentre{\splxs}$, let $\ell$ be
  any line through $c$ and orthogonal to $\affhull{\splxs}$. There
  must be a point $\hat{c} \in \ell$ such that the circumscribing ball
  for $\splxs$ centred at $\hat{c}$ is not empty. If this were not the
  case, we would have $\affhull{\splxs} = \affhull{\pts}$, and
  thus $\dim \affhull{\pts} < m$. It follows then (from the continuity
  of the radius of the circumballs parameterized by $\ell$), that
  there is a point $c' \in \simplex{c,\hat{c}}$ that is the centre of
  a Delaunay ball for a simplex $\splxs'$ that has $\splxs$ as a
  proper face. The first assertion follows.

  The second assertion follows from the same argument, and the
  observation that if $\splxs$ is not on the boundary of
  $\convhull{\pts}$, then there must be non-empty balls centred on
  $\ell$ at either side of $c$. If $p \in \pts \setminus
  \affhull{\splxs}$ is on the boundary of an empty ball centred at one
  side of $c$, by the intersection properties of spheres, it cannot be
  on the boundary of an empty ball centred on the other side of
  $c$. Thus there must be at least two distinct Delaunay
  $(k+1)$-simplices that share $\splxs$ as a face.
\end{proof}

\Lemref{lem:maximal.splx} gives rise to the following observation, which
plays an important role in \Secref{sec:protect.thick}, where we argue
that protecting the Delaunay $m$-simplices yields a thickness bound on
the simplices.
\begin{lem}[Separation]
  \label{lem:separate}
  If $\splxt \in \delP$ is a $j$-simplex that is not a boundary
  simplex, and $q \in \pts \setminus \splxt$, then there is a Delaunay
  $m$-simplex $\splx{m}$ which has $\splxt$ as a face, but does not
  include $q$.
\end{lem}
\begin{proof}
  Assume $j<m$, for otherwise there is nothing to prove.  If $\splxs =
  \splxjoin{q}{\splxt}$ is not Delaunay, the assertion follows from
  the first part of \Lemref{lem:maximal.splx}. Assume $\splxs$ is
  Delaunay and let $\tsplx{m}$ be a Delaunay $m$-simplex that has
  $\splxs$ as a face. Thus $\tsplx{m} = \splxjoin{q}{\splx{m-1}}$ for
  some Delaunay $(m-1)$-simplex, $\splx{m-1}$. Since $\splxt \leq
  \splx{m-1}$ does not belong to the boundary of $\convhull{\pts}$,
  neither does $\splx{m-1}$, so by the second part of
  \Lemref{lem:maximal.splx}, there is another Delaunay $m$-simplex
  $\splx{m}$ that has $\splx{m-1}$ (and therefore $\splxt$) as a
  face. Since $\splx{m}$ is distinct from $\tsplx{m}$, it does not
  have $q$ as a vertex.
\end{proof}

\subsubsection{The Delaunay complex in other metrics}
\label{sec:Delaunay.alt.metric}

We will also consider the Delaunay complex defined with respect to a
metric $\gdist$ on $\rem$ which differs from the Euclidean
one. Specifically, if $\pts \subset U \subset \rem$ and $\gdist: U
\times U \to \reel$ is a metric, then we define the Delaunay complex
$\delPd$ with respect to the metric $\gdist$.

The definitions are exactly analogous to the Euclidean case: A
Delaunay ball is a maximal empty ball $\ball{x}{r}$ in the metric
$\gdist$. The resulting Delaunay complex $\delPd$ consists of all the
simplices which are circumscribed by a Delaunay ball with respect to
the metric $\gdist$. The simplices of $\delPd$ are, possibly
degenerate, geometric simplices in $\rem$. As for $\delP$, $\delPd$
has the combinatorial structure of an abstract simplicial complex, but
unlike $\delP$, $\delPd$ may fail to be embedded even
when there are no degenerate simplices.


%

\subsection{Protection}
\label{sec:protection}

A Delaunay simplex $\splxs$ is \defn{$\delta$-protected} if it has a
Delaunay ball $B$ such that $\distEm{q}{\bdry{B}} > \delta$ for all
$q \in \pts \setminus \splxs$.  We say that $B$ is a
$\delta$-protected Delaunay ball for $\splxs$. If $\splxt < \splxs$,
then $B$ is also a Delaunay ball for $\splxt$, but it cannot be a
$\delta$-protected Delaunay ball for $\splxt$.
We say that $\splxs$ is \defn{protected} to mean that it is
$\delta$-protected for some unspecified $\delta \geq 0$.
\begin{figure}[ht]
  \begin{center}
    \includegraphics[width=.4\columnwidth]{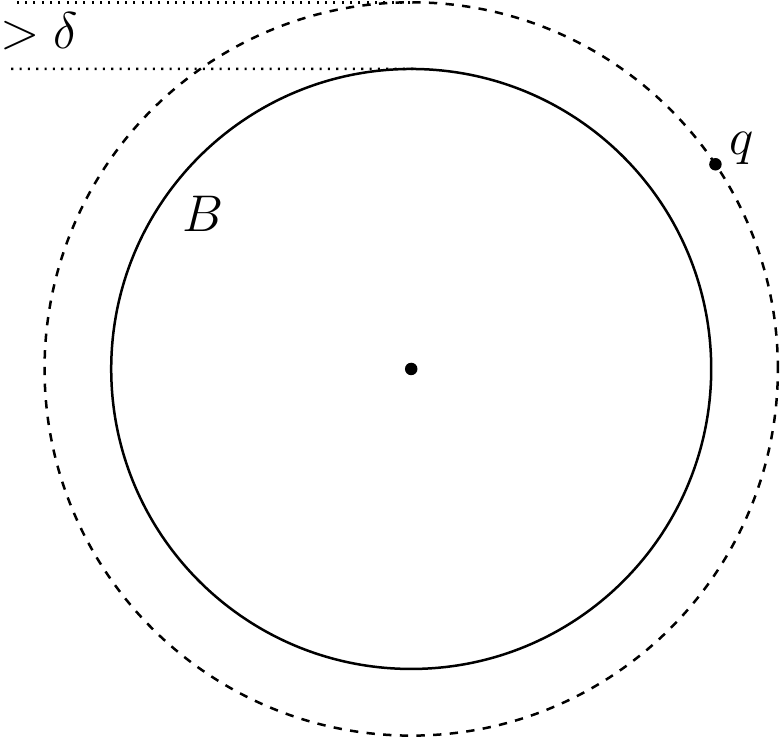} 
  \end{center}
  \caption{A Delaunay simplex $\splxs$ is $\delta$-protected if it has a
    Delaunay ball $\ballEm{c}{r}$ such that $\cballEm{c}{r + \delta}
    \cap (\pts \setminus \splxs) = \emptyset$.}
  \label{fig:protected.ball}
\end{figure}
\begin{de}[$\delta$-generic]
  \label{def:delta.generic}
  A finite set of points $\pts \subset \rem$ is
  \defn{$\delta$-generic} if $\affhull{\pts} = \rem$, and all the
  Delaunay $m$-simplices are $\delta$-pro\-tec\-ted. The set $\pts$ is
  simply \defn{generic} if it is $\delta$-generic for some unspecified
  $\delta \geq 0$.
\end{de}
Observe that we have employed a strict inequality in the definition of
$\delta$-protection. In particular, a $\delta$-generic point set is
generic even when $\delta=0$. In order for the quantity $\delta$ to be
meaningful, it should be considered with respect to a sampling radius
$\samconst$ for $\pts$. We will always assume that $\delta \leq
\samconst$.

In his seminal work, Delaunay~\cite{delaunay1934} demonstrated that if
there is no empty ball with $m+2$ points from $\pts$ on its boundary,
then $\delP$ is realized as a simplicial complex in $\rem$. In other
words $\delP$ is an embedded complex, and in fact it is a
triangulation of $\pts$, the \defn{Delaunay triangulation}. If $\pts$
is generic according to \Defref{def:delta.generic}, then Delaunay's
criterion will be met. This is obvious if there are no degenerate
$m$-simplices, and \Defref{def:delta.generic} ensures that a
degenerate $m$-simplex cannot exist in $\delP$, as shown by
\Lemref{lem:embed.del.star} below.

In particular, if $\pts$ is generic if and only if there are no Delaunay
simplices with dimension higher than $m$.  We can say more. There are
no degenerate Delaunay simplices.
This can be inferred directly from Delaunay's
result~\cite{delaunay1934}, but is also easily established from
\Lemref{lem:maximal.splx}. In \Secref{sec:protect.thick} we will
quantify this observation with a bound on the thickness of the
Delaunay simplices.

The $\delta$-generic assumption means that all the Delaunay
$m$-simplices are $\delta$-protected, but the lower dimensional
Delaunay do not necessarily enjoy this level of protection. The fact
that there are no degenerate Delaunay simplices implies that all the
simplices of all dimensions are $\tdelta$-protected for some $\tdelta
> 0$. 

\subsubsection{Local Delaunay triangulation}
\label{sec:local.del.tri}

Delaunay avoided boundary complications by assuming a periodic point
set, but we are particularly interested in the case where the point
sets come from local patches of a well-sampled compact manifold
without boundary. Periodic boundary conditions are not appropriate in
this setting, but this is not a problem because, as we show here,
Delaunay's argument applies locally.

Delaunay's proof that the Delaunay complex of a generic periodic point
set is a triangulation of $\rem$ consists of two observations. First
it is observed that if two Delaunay simplices intersect, then they
intersect in a common face. This shows that $\delP$ is embedded. The
argument is not complicated by the presence of boundary points:
\begin{lem}[Embedded star]
  \label{lem:embed.del.star}
  Suppose $\affhull{\pts} = \rem$ and $p \in \pts$. If all the
  $m$-simplices in $\str{p;\delP}$ are protected, then $\str{p;\delP}$ is
  embedded, and it is a pure $m$-complex.
\end{lem}
\begin{proof}
  We first observe that the $m$-simplices in $\str{p}$ are not
  degenerate.  If $\splxs^m$ is degenerate, then by
  \Lemref{lem:maximal.splx}, there is a simplex $\splxt$ with
  $\affhull{\splxt} = \rem$, and $\splxs^m < \splxt$. We have $\splxt
  \in \str{p}$, since $p \in \splxt$. An affinely independent set of
  $m+1$ vertices from $\splxt$ defines a non-degenerate $m$-simplex
  $\tsplxs^m < \splxt$, and since its unique circumball is also a
  Delaunay ball for $\splxt$, it cannot be protected, a contradiction.

  Now suppose that $\splxs, \splxt \in \str{p}$ and $\splxs \cap
  \splxt \neq \emptyset$. We need to show that they intersect in a
  common face.  By \Lemref{lem:maximal.splx}, we may assume that
  $\splxs$ and $\splxt$ are $m$-simplices.  Assume $\splxs \neq
  \splxt$, and let $B_1$ and $B_2$ be the Delaunay balls for $\splxs$
  and $\splxt$. Then $\affhull{\bdry{B_1} \cap \bdry{B_2}}$ defines an
  $(m-1)$-flat, $H$. Since $B_1$ and $B_2$ are empty balls, $H$
  separates the interiors of $\splxs$ and $\splxt$, and thus they must
  intersect in $H$, i.e., at the common face defined by the vertices
  in $\bdry{B_1} \cap \bdry{B_2}$.
\end{proof}

The second observation Delaunay made is that, in the case of a
periodic (infinite) point set, every $(m-1)$-simplex is the face of
two $m$-simplices (\Lemref{lem:maximal.splx}). The implication here is
that $\delP$ cannot have a boundary, and therefore must cover
$\rem$. Here we flesh out the argument for our purposes: If an
embedded finite complex contains $m$-simplices then its topological
boundary must contain $(m-1)$-simplices. We first observe that the
topological boundary of an embedded complex is defined by a
subcomplex:
\begin{lem}[Boundary complex]
  \label{lem:bdry.complex}
  If $K$ is an embedded (finite) complex in $\rem$, then the topological
  boundary of $\carrier{K} \subset \rem$ is defined by a subcomplex:
  $\bdry{\carrier{K}} = \carrier{\bd{K}}$, where the subcomplex
  $\bd{K} \subset K$ is called the \defn{boundary complex} of $K$.
\end{lem}
\begin{figure}[ht]
  \begin{center}
    \includegraphics[width=.6\columnwidth]{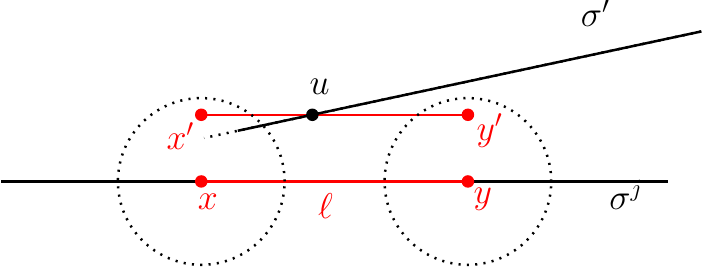} 
  \end{center}
  \caption{Diagram for the proof of \Lemref{lem:bdry.complex}.}
  \label{fig:bdry.cplx}
\end{figure}
%
\begin{proof}
  Since $K$ is finite, $\bdry{\carrier{K}}$ is contained in
  $\carrier{K}$.  Suppose $x \in \bdry{\carrier{K}}$. Then $x \in
  \intr{\splxs^j}$ for some $\splxs^j \in K$. We wish to show that
  $\splxs^j \subset \bdry{\carrier{K}}$. Suppose to the contrary that
  $y \in \intr{\splxs^j}$, but $y$ does not belong to
  $\bdry{\carrier{K}}$. This means that $y \in \intr{\carrier{K}}$.

  Consider the segment $\ell = \seg{x}{y} \subset
  \intr{\splxs^j}$. Let $Z \subset K$ be the subcomplex consisting of
  those simplices that do not contain $\splxs^j$. Let
  \begin{equation*}
    r_1 = \min_{\splxs \in Z} \distEm{\ell}{\splxs}.
  \end{equation*}
  Choosing $r \leq r_1$, and $x' \in \ballEm{x}{r} \setminus
  \carrier{K}$, let $y' = y + (x'-x)$. Since $y \in
  \intr{\carrier{K}}$, we may assume that $r$ is small enough so that
  $y' \in \intr{\carrier{K}}$.

  Consider the segment $\ell' = \seg{x'}{y'}$. By construction, $\ell'
  \cap \carrier{Z} = \emptyset$. However, consider the point $u \in
  \intr{\ell'}$ that is the point in $\ell' \cap \carrier{K}$ that is
  closest to $x'$. The point $u$ lies in the interior of some
  simplex $\splxs' \in K$, but we cannot have $\splxs^j \leq
  \splxs'$. Indeed if this were the case, $x'$ would lie in
  $\affhull{\splxs'}$, and so $u \in \bdry{\splxs'}$, contradicting
  the assumption that $u \in \intr{\splxs'}$.

  But this means that $\splxs' \in Z$, which contradicts the fact that
  $\ell' \cap \carrier{Z} = \emptyset$. Therefore we must have $y \in
  \bdry{\carrier{K}}$ for all $y \in \intr{\splxs^j}$.

  Finally, observe that if $\splxt < \splxs^j$, then $\splxt \subset
  \bdry{\carrier{K}}$, since $\bdry{\carrier{K}}$ is closed.
\end{proof}

\begin{lem}[Pure boundary complex]
  \label{lem:pure.bdry.complex}
  If $K$ is a (finite) pure $m$-complex embedded in $\rem$, then its
  boundary complex is a pure $(m-1)$-complex.
\end{lem}
\begin{figure}[ht]
  \begin{center}
    \includegraphics[width=.4\columnwidth]{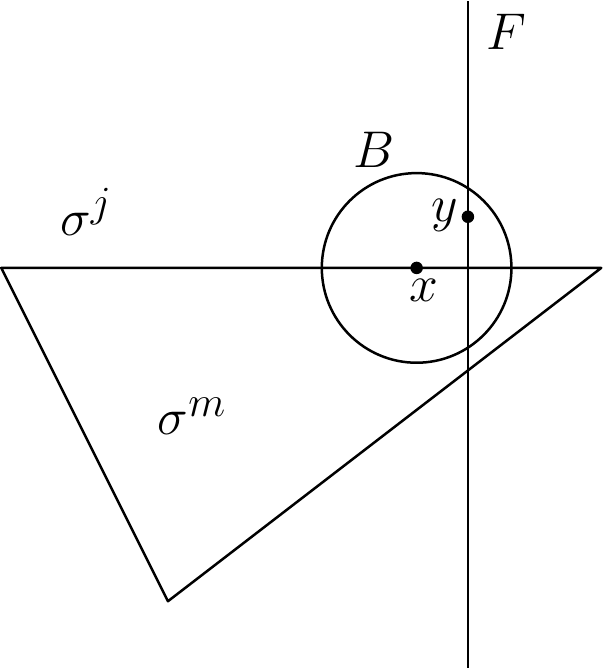} 
  \end{center}
  \caption{Diagram for the proof of \Lemref{lem:pure.bdry.complex}.}
  \label{fig:pure.bdry}
\end{figure}
\begin{proof}
%
  Since $K$ is finite, $\bd{K}$ is nonempty; it contains at least the
  vertices in $\bdry{\convhull{\carrier{K}}}$. We will show that if
  $\splxs^j \in \bd{K}$, is a $j$-simplex, with $0 \leq j < m-1$, then
  there is a $\splxs^k \in \bd{K}$ with $\splxs^j < \splxs^k$. The
  result then follows, since $\bd{K}$ cannot contain $m$-simplices,
  because $K$ is embedded.

  Suppose $\splxs^j \in \bd{K}$, and $x \in \intr{\splxs^j}$. Let $Z
  \subset K$ be the subcomplex consisting of simplices that do not
  contain $\splxs^j$, and let
  \begin{equation*}
    r = \min_{\splxs \in Z} \distEm{x}{\splxs}.
  \end{equation*}
  Let $B = \ballEm{x}{r}$, and choose $y \in B \setminus
  \carrier{K}$. Let $F$ be the $(m-j)$-dimensional affine space
  orthogonal to $\affhull{\splxs^j}$ and containing $y$, and let
  $S^{m-j-1} = F \cap \bdry{\ballEm{x}{r'}}$, where $r' =
  \norm{x-y}$. See \Figref{fig:pure.bdry}.

  Since $K$ is pure, there is an $m$-simplex $\splxs^m$ with $\splxs^j
  < \splxs^m$. We have $\splxs^m \cap S^{m-j-1} \neq
  \emptyset$. Indeed, choose $w \in \intr{\splxs^m}$, and $u \in
  \splxs^j$ different from $x$, and observe that the plane $Q$ defined
  by $x,w,u$ intersects $B \cap \splxs^m$ in a semi-disk, by
  construction of $B$. By the construction of $S^{m-j-1}$, it must
  intersect this semidisk.

  Let $z \in S^{m-j-1}$ be a point that minimises the geodesic
  distance in $S^{m-j-1}$ to $y$. Then $z \in
  \bdry{\carrier{K}}$. Thus $z \in \intr{\splxs^k}$ for some $\splxs^k
  \in \bd{K}$, and since $z \in B$, $\splxs^k$ cannot belong to
  $Z$. Thus $\splxs^j \leq \splxs^k$, but since $S^{m-j-1} \cap
  \splxs^j = \emptyset$, we have $\splxs^j < \splxs^k$.
\end{proof}

From \Lemref{lem:maximal.splx}, \Lemref{lem:embed.del.star}, and
\Lemref{lem:pure.bdry.complex}, one can verify that if $\pts$ is
generic then $\bdry{\convhull{\pts}} = \bdry{\carrier{\delP}}$, and
thus obtain the standard result that $\delP$ is a triangulation of
$\pts$. However, we are interested localizing the result, without the
assumption that the entire point set is generic.  We have the
following local version of Delaunay's triangulation result:
\begin{lem}[Local Delaunay triangulation]
  \label{lem:Delaunay.local.tri}
  If $p \in \pts$ is an interior point, and the Delaunay $m$-simplices
  incident to $p$ are protected, then $\delP$ is a triangulation at $p$.
\end{lem}
\begin{proof}
  By \Lemref{lem:embed.del.star}, $\str{p}$ is a pure $m$-complex, and
  it is embedded. It follows then from \Lemref{lem:pure.bdry.complex},
  that the boundary complex $\bd{\str{p}}$ is a pure
  $(m-1)$-complex. Thus $p$ cannot belong to $\bd{\str{p}}$. Indeed,
  it follows from \Lemref{lem:maximal.splx} that any $(m-1)$-simplex
  $\splxs \in \str{p}$ is the face of at least two $m$-simplices in
  $\delP$, and if $p \in \splxs$, then both of these $m$-simplices
  belong to $\str{p}$, and are embedded, with intersection
  $\splxs$. Thus $p$ cannot belong to an $(m-1)$-simplex in
  $\bd{\str{p}}$, and therefore $p \in \intr{\carrier{\str{p}}}$.

  It remains to verify Condition~\ref{tri:no.conflicts} of
  \Defref{def:local.triangulation}. The argument is similar to the
  proof of \Lemref{lem:embed.del.star}:
  Suppose $x \in (\intr{\splxt}) \cap \splxs$ for $\splxs \in
  \str{p}$. We may assume that $\splxs$ is an $m$-simplex. Then
  consider the Delaunay balls $B_1$ for $\splxs$ and $B_2$ for
  $\splxt$. If $B_1 = B_2$, then, since $\splxs$ is protected,
  $\splxt$ must be a face of $\splxs$, and so belong to
  $\str{p}$. Assume then that $B_1 \neq B_2$, and let $H$ be the
  $(m-1)$-flat defined by $\affhull{\bdry{B_1} \cap \bdry{B_2}}$.
  Since $B_1$ is empty, $x \in \intr{\splxt}$ cannot lie in the open
  half-space defined by $H$ and containing $\splxs$. Since $x \in
  \splxs$ also, it must lie in $H$, and therefore all vertices of
  $\splxt$ lie in $H \cap \bdry{B_2} = H \cap \bdry{B_1}$, and so
  $\splxt$ is a face of $\splxs$.
\end{proof}

\subsubsection{Safe interior simplices}

We wish to consider the properties of Delaunay triangulations in
regions which are comfortably in the interior of $\convhull{\pts}$,
and avoid the complications that arise as we approach the boundary of
the point set. We introduce some terminology to facilitate this.

If none of the vertices of $\splxs$ lie on $\bdry{\convhull{P}}$, then
it is an \defn{interior simplex}.  We wish to identify a subcomplex of
the interior simplices of $\delP$ consisting of those simplices whose
neighbour simplices are also all interior simplices with small
circumradius. An interior simplex near the boundary of
$\convhull{\pts}$ does not necessarily have its circumradius
constrained by the sampling radius. However, we have the following:
\begin{lem}
  \label{lem:interior.splx}
  If $\pts$ is an $\samconst$-sample set, and $\splxs \in \delP$ has a
  vertex $p$ such that $\distEm{p}{\bdry{\convhull{\pts}}} \geq
  2\samconst$, then $\circrad{\splxs} < \samconst$ and $\splxs$ is an
  interior simplex.
\end{lem}
\begin{proof}
  Let $\ballEm{c}{r}$ be a Delaunay ball for $\splxs$. We will show $r
  < \samconst$. Suppose to the contrary. Let $x$ be the point on
  $\seg{c}{p}$ such that $\distEm{p}{x} = \samconst$. Then $p$ is the
  closest point in $\pts$ to $x$, and so the sampling criteria imply
  that $\distEm{x}{\bdry{\convhull{\pts}}} < \samconst$. But then
  $\distEm{p}{\bdry{\convhull{\pts}}} \leq \distEm{p}{x} +
  \distEm{x}{\bdry{\convhull{\pts}}} < 2\samconst$, contradicting the
  hypothesis on $p$.

  Thus $r < \samconst$, and it follows that $\splxs$ is an interior
  simplex because if $q \in \splxs$, then $\distEm{p}{q} \leq 2r <
  \distEm{p}{\bdry{\convhull{\pts}}}$. 
\end{proof}

This suggests the following:
\begin{de}[Deep interior points]
  Suppose $\pts \subset \rem$ is an $\samconst$-sample set.  The
  subset $\dipts \subset \pts$ consisting of all $p \in \pts$ with
  $\distEm{p}{\bdry{\convhull{\pts}}} \geq 4\samconst$ is the set of
  \defn{deep interior points}.
\end{de}
By \Lemref{lem:interior.splx}, all the simplices that include a deep
interior point, as well as all the neighbours of such simplices, will
have a small circumradius. For technical reasons it is inconvenient to
demand that \emph{all} the Delaunay $m$-simplices be
$\delta$-protected. We focus instead on a subset defined with respect
to a set of deep interior points:
\begin{figure}[ht]
  \begin{center}
    \includegraphics[width=.6\columnwidth]{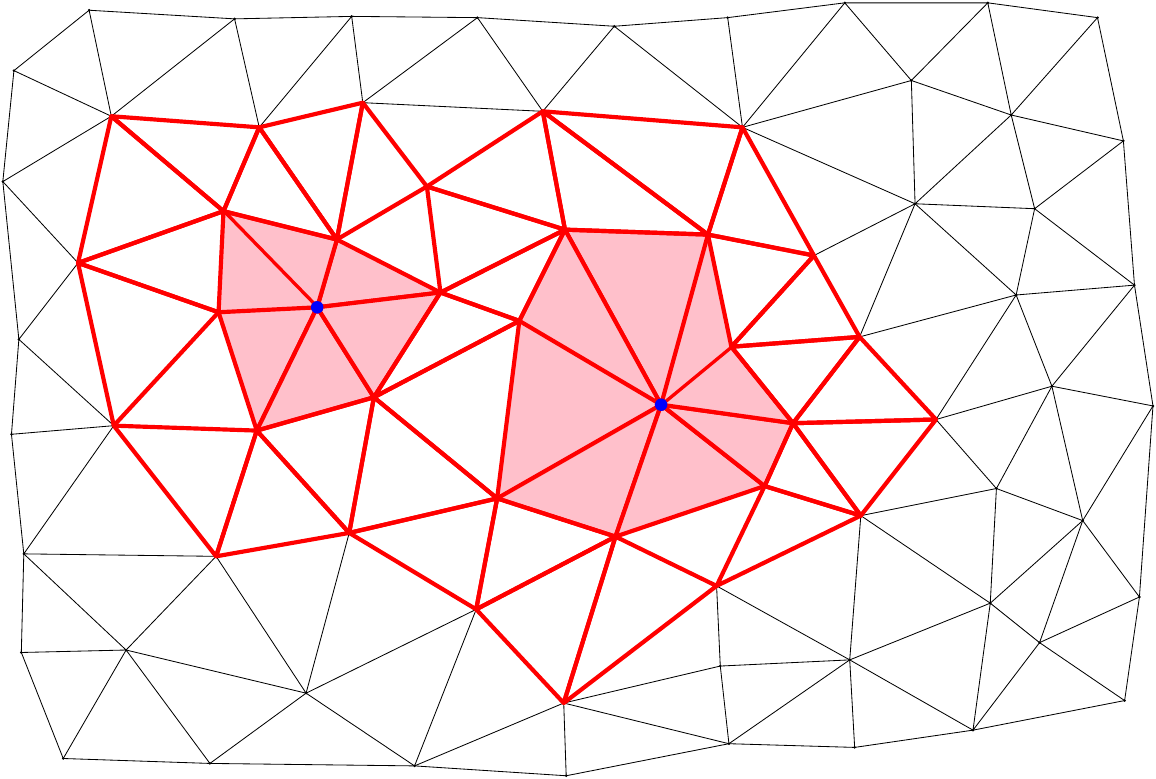} 
  \end{center}
  \caption{If $\pts$ is $\delta$-generic for $\sdipts$, then the safe
    interior simplices are the simplices in $\str{\sdipts}$. Here
    $\sdipts$ consists of the two large vertices (blue). They must be
    at least $4\samconst$ from $\bdry{\convhull{\pts}}$ (which is not
    depicted in the figure). The safe interior simplices are
    shaded. All the simplices in $\str{\str{\sdipts}}$ are
    $\delta$-protected. These simplices have bold outlines (red), but
    are not necessarily shaded. }
  \label{fig:safe.interior}
\end{figure}
\begin{de}[$\delta$-generic for $\sdipts$]
  \label{def:locally.delta.generic}
  The set $\pts \subset \rem$ is \defn{$\delta$-generic for $\sdipts$}
  if $\sdipts \subseteq \dipts$ and all the $m$-simplices in
  $\str{\str{\sdipts; \delP}}$ are $\delta$-protected.  The \defn{safe
    interior simplices} are the simplices in $\str{\sdipts ;\delP}$.
\end{de}
Thus the safe interior simplices are determined by our choice of
$\sdipts \subseteq \dipts$, and our protection requirements ensure
that all the $m$-simplices that share a face with a safe interior
simplex are $\delta$-protected and have a small circumradius. A
schematic example is depicted in \Figref{fig:safe.interior}.

\subsection{Thickness from protection}
\label{sec:protect.thick}

Our goal here is to demonstrate that the safe interior simplices on a
$\delta$-generic point set are $\thickbnd$-thick. If $\delta =
\protconst \samconst$, for some constant $\protconst \leq 1$, then we
obtain a constant $\thickbnd$ which depends only on $\protconst$.
The key observation is that together with \Lemref{lem:separate},
protection imposes constraints on all the Delaunay simplices;
they cannot be too close to being degenerate.  In the particular case
that $j=0$, \Lemref{lem:separate} immediately implies that the
vertices of the safe interior simplices are $\delta$-separated:
\begin{lem}[Separation from protection]
  \label{lem:delta.sparse}
  If $\pts$ is $\delta$-generic for $\sdipts$, then
  $\shortedge{\splxs} > \delta$ for any safe interior simplex
  $\splxs$. 
\end{lem}

\begin{figure}[th]
  \begin{center}
    \subfloat[$H$ separates $c$ and $q$]{ 
      \label{sfig:opposite.centres}
      \includegraphics[width=.3\columnwidth]{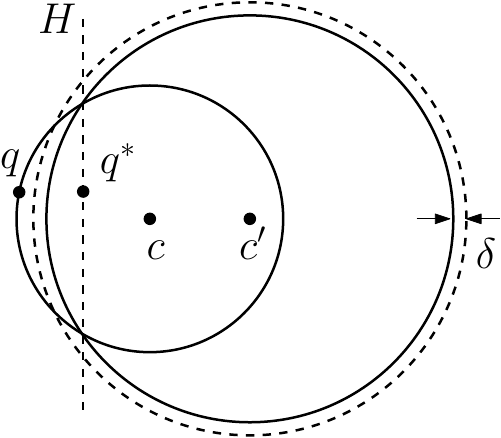} }
    \subfloat[$q$ and $c$ on same side of $H$]{
      \label{sfig:q.with.centre}
      \includegraphics[width=.6\columnwidth]{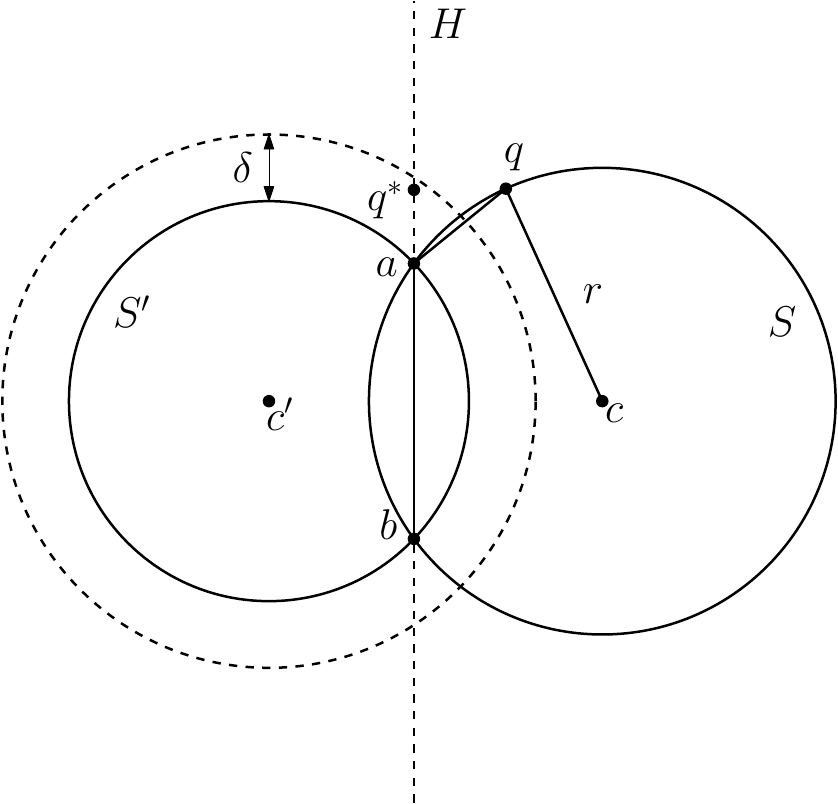} }
  \end{center}
  \caption{Diagram for \Lemref{lem:bound.height}. \protect\subref{sfig:opposite.centres} When $H$ separates $q$ and $c$ then $\distEd{q}{q^*} > \delta$.  \protect\subref{sfig:q.with.centre} Otherwise, a lower bound on the distance between $q$ and its projection $q^*$ on $H$ is obtained by an upper bound on the angle $\angle qab$.}
  \label{fig:bound.height}
\end{figure}
\begin{lem}
  \label{lem:bound.height}
  Suppose that $B = \ballEm{c}{r}$ is a Delaunay ball for $\splxs =
  \splxjoin{q}{\splxt}$ with $r<\samconst$ and that
  $\shortedge{\splxt} \geq \sparsity$ for some $\sparsity \leq
  \samconst$. Suppose also that $\splxt \leq \splxs'$ and that
  $\splxs$ is not a face of $\splxs'$. 

  If $B'$ is a $\delta$-protected Delaunay ball for $\splxs'$, and $H
  = \affhull{\bdry{B} \cap \bdry{B'}}$, then
  \begin{equation*}
    \distEm{q}{H} >
    \frac{\sqrt{3}\delta}{4\samconst}(\sparsity + \delta).
  \end{equation*}

  It follows that, if $\pts$ is $\delta$-generic for $\sdipts$, with
  sampling radius $\samconst$, and $\splxt$ is a safe interior
  simplex, then
  \begin{equation*}
    \splxalt{q}{\splxs} > \frac{\sqrt{3}\delta^2}{2\samconst}.
  \end{equation*}
\end{lem}
\begin{proof}
  Let $B' = \ballEm{c'}{r'}$ be the $\delta$-protected Delaunay ball
  for $\splxs'$. Our geometry will be performed in the plane, $Q$,
  defined by $c$, $c'$, and $q$. This plane is orthogonal to the
  $(m-1)$-flat $H$, and it follows that the distance $\distEm{q}{H}$
  is realized by a segment in the plane $Q$: the projection, $q^*$, of
  $q$ onto $H$ lies in $Q$, and $\distEm{q}{H} = \distEm{q}{q^*}$.

  If $H$ separates $q$ from $c$, then $\bdry{B'}$ separates $q$ from
  $q^*$, and $\distEm{q}{q^*} > \distEm{q}{\bdry{B'}} > \delta$, since
  $B'$ is $\delta$-protected
  (\Figref{fig:bound.height}\subref{sfig:opposite.centres}). The lemma
  then follows since $\sparsity$ and $\delta$ are each no larger than
  $\epsilon$.  Thus assume that $q$ and $c$ lie on the same side of
  $H$, as shown in
  \Figref{fig:bound.height}\subref{sfig:q.with.centre}. Let $S' = Q
  \cap \bdry{B'}$, and $S = Q \cap \bdry{B}$, and let $a$ and $b$ be
  the points of intersection $S' \cap S$. Thus $H \cap Q$ is the line
  through $a$ and $b$.

  We will bound $\distEm{q}{q^*}$ by finding an upper bound on the
  angle $\gamma = \angle qab$. This is the same as the standard
  calculation for upper-bounding the angles in a triangle with bounded
  circumradius to shortest edge ratio. Without loss of generality, we
  may assume that $\gamma \geq \angle qba$, and we will assume that
  $\gamma \geq \pi/2$ since otherwise $q* \in B'$ and thus
  $\distEm{q}{q*} > \delta$ and the lemma is again trivially
  satisified.

  Since $\distEm{a}{q} > \delta$, we have $\distEm{q}{q^*} =
  \distEm{a}{q}\sin \gamma > \delta \sin \gamma$. Also note that
  $\dist{a}{b} \geq 2\circrad{\splxt} \geq \shortedge{\splxt} \geq
  \sparsity$.  Let $\alpha = \angle qac$. Then $\cos \alpha =
  \frac{\distEm{a}{q}}{2r} \geq \frac{\delta}{2\samconst}$, which
  means that $\alpha
  \leq \arccos \frac{\delta}{2\samconst} \leq \frac{\pi}{2}$.  Similarly,
  with $\beta = \angle cab$, we have $\beta \leq \arccos
  \frac{\sparsity}{2\samconst} \leq \frac{\pi}{2}$. Thus $\frac{\pi}{2}
  \leq \gamma = \alpha + \beta \leq \gamma'$, where
  \begin{equation*}
    \gamma' = \arccos \frac{\delta}{2\samconst} + \arccos
    \frac{\sparsity}{2\samconst}.
  \end{equation*}
  Since $\sin \gamma \geq \sin \gamma'$, when $\frac{\pi}{2} \leq
  \gamma \leq \gamma' \leq \pi$, we have
  \begin{equation*}
    \begin{split}
      \distEm{q}{q^*} &> \delta \sin \gamma  \geq \delta \sin \gamma'\\
      &= \delta \sin \left( \arccos
      \frac{\delta}{2\samconst} + \arccos
      \frac{\sparsity}{2\samconst} \right)\\ 
      &\geq \delta \left(\frac{\sparsity}{2\samconst}\sin \left(\arccos
      \frac{\delta}{2\samconst} \right)  +
      \frac{\delta}{2\samconst}\sin \left( \arccos
        \frac{\sparsity}{2\samconst} \right) \right)\\
      &\geq \frac{\sqrt{3}\delta}{4\samconst}(\sparsity +
      \delta), 
    \end{split}
  \end{equation*}
  where the last inequality follows from $\sparsity \leq \samconst$
  and $\delta \leq \samconst$.

  Since $\affhull{\splxt} \subset H$, it follows that
  $\splxalt{q}{\splxs} \geq \distEm{q}{H}$, and if $\pts$ is
  $\delta$-generic for $\sdipts$, then $\sparsity \geq \delta$, and
  \Lemref{lem:separate} ensures that there is a $\delta$-protected
  $\splxs'$ that contains $\splxt$ but not $q$.
\end{proof}

We thus obtain a bound on the thickness of the safe interior simplices
when $\pts$ is $\delta$-generic for $\sdipts$. Since
\Lemref{lem:bound.height} yields a lower bound of
$\frac{\sqrt{3}\delta^2}{2\samconst}$ on the altitudes of the safe
interior simplices, and since $\longedge{\splxs} \leq 2\samconst$, we
have that $\thickness{\splxs} \geq
\frac{\sqrt{3}\delta^2m}{4\samconst^2}$ for all safe interior
$\splxs$. If $\delta = \protconst \samconst$, we obtain a constant
thickness bound.
\begin{thm}[Thickness from protection]
  \label{thm:prot.thick}
  If $\pts \subset \rem$ is $\delta$-generic for $\sdipts$ with $\delta
  = \protconst \samconst$, where $\samconst$ is a sampling radius for
  $\pts$, then the safe interior simplices are $\thickbnd$-thick, with
  \begin{equation*}
    \thickbnd = \frac{\sqrt{3}\protconst^2}{4m}.
  \end{equation*}
\end{thm}

%

\section{Delaunay stability}
\label{sec:stability}

We find upper bounds on the magnitude of a perturbation for which a protected Delaunay ball remains a Delaunay ball. We consider both perturbations of the sample points in Euclidean space, and perturbations of the metric itself.  The primary technical challenge is bounding the effect of a perturbation on the circumcentre of an $m$-simplex. We then find the relationship between the perturbation parameter $\pertconst$ and the protection parameter $\delta$ which ensures that a $\delta$-protected Delaunay simplex will remain a Delaunay simplex.

\subsection{Perturbations and circumcentres}

As expected, a bound on the displacement of the circumcentre requires
a bound on the thickness of the simplex.

\subsubsection{Almost circumcentres}

If $\splxs$ is thick, a point whose distances to the vertices of $\splxs$ are all almost the same, will lie close to $\normhull{\splxs}$. 
\begin{lem}
  \label{lem:power.close.centres}
  If $\splxs = \simplex{p_0,...,p_k} \subset \rem$ is a
  non-degenerate $k$-simplex, and $x \in \rem$ is such that
  \begin{equation}
    \label{eq:sqr.nr.ctr}
    \abs{\norm{p_i-x}^2 - \norm{p_j-x}^2} \leq \trelconst^2 \quad
    \text{ for all } i,j \in [0,...,k],
  \end{equation}
  then there is a point $c \in \normhull{\splxs}$ such that $\norm{c-x} \leq
  \localconst$, where
  \begin{equation*}
    \localconst =
    \frac{\trelconst^2}{2\thickness{\splxs}\longedge{\splxs}}. 
  \end{equation*}
  In particular, if $\splxs$ is an $m$-simplex then $x \in
  \cballEm{\circcentre{\splxs}}{\localconst}$.

  If the inequalities in Equations~\eqref{eq:sqr.nr.ctr} are made
  strict, then the conclusions may also be stated with strict
  inequalities.
\end{lem}
\begin{proof}
  First suppose $k=m$. The circumcentre of $\splxs$
  is given by the linear equations $\norm{\circcentre{\splxs} -
    p_i}^2 = \norm{\circcentre{\splxs} - p_0}^2$, or
  \begin{equation*}
    \transp{(p_i - p_0)}\circcentre{\splxs} = \frac{1}{2}(\norm{p_i}^2
    - \norm{p_0}^2).
  \end{equation*}
  Letting $b$ be the vector whose $i^{\text{th}}$ component is defined
  by the right hand side of the equation, and letting $P$ be the $m
  \times m$ matrix, whose $i^\text{th}$ column is $(p_{i} - p_{0})$,
  we write the equations in matrix form:
  \begin{equation}
    \label{eq:cc}
    \transp{P}\circcentre{\splxs} = b.
  \end{equation}

  Without loss of generality, assume $p_0$ is the vertex that
  minimizes the distance to $x$. Then, defining $x_a$ to be the vector
  whose $i^{\text{th}}$ component is $\frac{1}{2}(\norm{p_i - x}^2 -
  \norm{p_0 - x}^2)$, we have $\norm{p_i-x}^2 = \norm{p_0-x}^2 +
  2(x_a)_i$, and we find
  \begin{equation}
    \label{eq:almost.cc}
    \transp{P}x = b - x_a.
  \end{equation}

  From Equations \eqref{eq:cc} and \eqref{eq:almost.cc} we have
  \begin{equation*}
    \norm{\circcentre{\splxs} - x} = \norm{\inv{(\transp{P})}x_a} \leq
    \norm{\inv{P}}\norm{x_a}. 
  \end{equation*}
  From \Eqnref{eq:sqr.nr.ctr}, it follows that $\norm{x_a} \leq
  \frac{\sqrt{m}\trelconst^2}{2}$, and from Lemmas
  \ref{lem:svd.pseudo.inv} and \ref{lem:bound.skP} we have
  $\norm{\inv{P}} \leq
  (\sqrt{m}\thickness{\splxs}\longedge{\splxs})^{-1}$, and thus the
  result holds for full dimensional simplices.

  If $\splxs$ is a $k$-simplex with $k \leq m$, then we consider
  $\hat{x}$, the orthogonal projection of $x$ into
  $\affhull{\splxs}$. We observe that $\hat{x}$ also must satisfy
  \Eqnref{eq:sqr.nr.ctr}, and we conclude from the above argument that
  $\norm{\circcentre{\splxs} - \hat{x}} \leq \localconst$. Then
  letting $c = \circcentre{\splxs} + (x - \hat{x})$ we have that $c
  \in \normhull{\splxs}$ and $\norm{c-x} \leq \localconst$.
\end{proof}

It will be convenient to have a name for a point that is almost
equidistant to the vertices of a simplex:
\begin{de}
  \label{def:rel.centre}
  A \defn{$\relconst$-centre} for a simplex $\splxs =
  \simplex{p_0,\ldots ,p_k} \subset \rem$ is a point $x$ that
  satisfies
  \begin{equation}
    \label{eq:relaxed.cc}
    \big| \norm{p_i - x} - \norm{p_j -x} \big| \leq \relconst \qquad \text{
      for all } i,j \leq k.
  \end{equation}
\end{de}
With a bound on the distance from $x$ to the vertices of $\splxs$,
\Lemref{lem:power.close.centres} can be transformed into a bound on
the distance from a $\relconst$-centre to the closest true centre in
$\normhull{\splxs}$:
\begin{lem}
  \label{lem:close.to.centres}
  If $\splxs = \simplex{p_0,\ldots ,p_k} \subset \rem$ is
  non-degenerate, and for some $\relconst > 0$ the point $x \in \rem$
  is a $\relconst$-centre such that
  \begin{equation*}
    \norm{p_i - x} < \tsamconst \qquad \text{ for all } i,j \leq k,
  \end{equation*}
  then there exists a $c \in \normhull{\splxs}$ such that $\norm{x-c}
  < \localconst$, where
  \begin{equation*}
    \localconst = 
    \frac{\tsamconst\relconst}{\thickness{\splxs}\longedge{\splxs}}. 
  \end{equation*}
  In particular, if $\splxs$ is an $m$-simplex, then $x \in
  \ballEm{\circcentre{\splxs}}{\localconst}$. 
\end{lem}
\begin{proof}
  Let $R = \max_i \norm{p_i -x}$ and $r = \min_i \norm{p_i-x}$. Then
  \begin{equation*}
    R^2 - r^2 = (R+r)(R-r) < 2\tsamconst(R-r) \leq
    2\tsamconst\relconst,
  \end{equation*}
  and the result then follows from \Lemref{lem:power.close.centres}.
\end{proof}


%
%
%

\subsubsection{Circumcentres and metric perturbations}
\label{sec:metric.cc.pert}

We will show here that for an $\thickbnd$-thick $m$-simplex $\splxs$
in $\rem$, and a metric $\gdist$ that is close to $\gdistEm$, there
will be a metric circumcentre $c$ near $\circcentre{\splxs}$.
We require the metric $\gdist$ to be continuous in the topology
defined by $\gdistEm$. 
Henceforth, whenever we refer to a \defn{perturbation of the Euclidean
  metric}, this topological compatibility will always be assumed. 

\begin{figure}[ht]
  \begin{center}
    \includegraphics[width=.6\columnwidth]{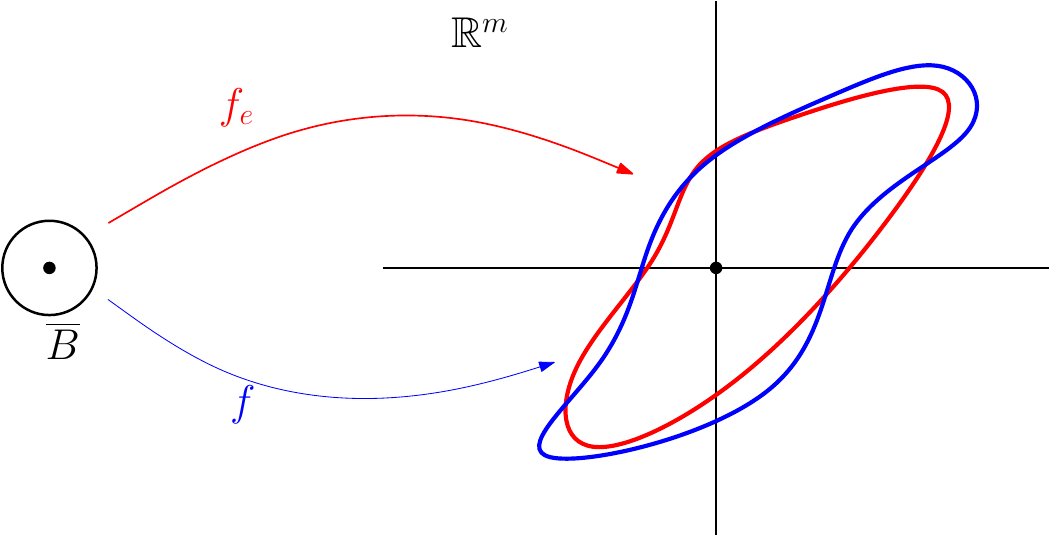} 
  \end{center}
  \caption{The maps $f_e$ and $f$ (described in the main text) map
    circumcentres to the origin. Since $f_e^{-1}(0)$ contains a unique
    point, and $f_e(\bdry{B})$ lies far from the origin, a
    consideration of the degree of the mappings, together with the
    fact that $f_e(\bdry{B})$ and $f(\bdry{B})$ are close, reveals
    that $f^{-1}(0)$ cannot be empty, and thus $B$ must contain a
    circumcentre of $\splxs^m$ with respect to the metric $\gdist$.}
  \label{fig:map.degree}
\end{figure}
The proof is a topological argument based on considering a mapping
into $\rem$ of a small ball around the circumcentre of $\splxs$. The
mapping is based on the metric and is such that circumcentres get
mapped to the origin. In the mapping associated to the Euclidean
metric, points that get mapped close to the origin are
$\relconst$-centres, and since the $\relconst$-centres are in the
interior of the ball, the boundary of the ball does not get mapped
close to the origin. A small perturbation of the metric yields a small
perturbation in the mapping, and so we can argue that there is a
homotopy between the mapping associated with the Euclidean metric and
the one associated to the perturbed metric, such that no point on the
boundary of the ball ever gets mapped to the origin. The situation is
depicted schematically in \Figref{fig:map.degree}. A consideration of
the degree of the mapping allows us to conclude that the ball must
contain a circumcentre for the perturbed metric.

We will demonstrate the following:
\begin{lem}[Circumcentres: metric perturbation]
  \label{lem:metric.cc.pert}
  Let $U \subset \rem$, and let $\gdist: U \times U \to \reel$ be a
  continuous metric with respect to the topology defined by $\gdistEm$,
  and such that for any $x,y \in U$ with $\distEm{x}{y} < 2\samconst$,
  we have $\abs{ \dist{x}{y} - \distEm{x}{y} } \leq \pertconst$, with
  \begin{equation*}
    \pertconst \leq \frac{\thickbnd \sparseconst
      \samconst}{8}. 
  \end{equation*}
  If $\splxs = \simplex{p_0,\ldots, p_m} \subset U$ is an
  $\thickbnd$-thick $m$-simplex with $\circrad{\splxs} < \samconst$,
  and $\shortedge{\splxs} \geq \sparseconst \samconst$, and such that
  $\distEm{p_i}{\bdry{U}} \geq 2\samconst$, then there is a point
  \begin{equation*}
    c \in B = \ballEm{\circcentre{\splxs}}{\localconst} \quad \text{
      with } \localconst = \frac{8 \pertconst}{\thickbnd \sparseconst},
  \end{equation*}
  and such that $\dist{c}{p_i}=\dist{c}{p_j}$ for all $p_i,p_j \in \splxs$.
\end{lem}

In order to prove \Lemref{lem:metric.cc.pert}, we will use a
particular case of \Lemref{lem:close.to.centres}:
\begin{lem}
  \label{lem:close.to.ctr}
  Suppose $\splxs$ is an $\thickbnd$-thick $m$-simplex such that
  $\shortedge{\splxs} \geq \sparseconst\samconst$. If $x$ is a
  $\relconst$-centre for $\splxs$ with $\distEm{x}{p} < 2\samconst$
  for all $p \in \splxs$, then $x \in
  \ballEm{\circcentre{\splxs}}{\localconst}$, where $\localconst =
  \frac{2\relconst}{\thickbnd \sparseconst}$.
\end{lem}
Let $B = \ballEm{\circcentre{\splxs}}{\localconst}$ be the open ball
which contains the $\relconst$-centres for $\splxs$.  We will show
that if $\relconst = 4\pertconst$, then a circumcentre $c$ for
$\splxs$ with respect to $\gdist$ will also lie in $B$.  However, we
make no claim that $c$ is unique.  Note that $\close{B} \subset U$.

Consider the function $f_e: \close{B} \to \rem$ given by
\begin{equation}
  \label{eq:euclidean.f}
  f_e(x) =
  \transp{(\distEm{x}{p_1} - \distEm{x}{p_0}, \ldots,
    \distEm{x}{p_m} - \distEm{x}{p_0})}.
\end{equation}
Observe that $f_e$ maps the circumcentre of $\splxs$, and only the
circumcentre, to the origin: $f_e^{-1}(0) = \{\circcentre{\splxs} \}$.

We construct a similar function for the metric $\gdist$, 
\begin{equation}
  \label{eq:metric.f}
  f(x) = \transp{(\dist{x}{p_1} - \dist{x}{p_0}, \ldots,
    \dist{x}{p_m} - \dist{x}{p_0})},
\end{equation}
and we will show that there must be a $c \in f^{-1}(0) \subset B$. We
first show that there is a homotopy between $f$ and $f_e$ such that
the image of $\bdry{\close{B}}$ never touches the origin: 
%
\begin{lem}
  \label{lem:fef.homotopic}
  Under the hypotheses of \Lemref{lem:metric.cc.pert},
  if $\relconst = 4\pertconst \leq \frac{\thickbnd \sparseconst
    \samconst}{2}$, then there is a homotopy $F: \close{B} \times
  [0,1] \to \rem$ between $f_e(x) = F(x,0)$ and $f(x) = F(x,1)$ with
  $F(x,t) \neq 0$ for all $x \in \bdry{\close{B}}$ and $t \in [0,1]$.
\end{lem}
\begin{proof}
  We define the homotopy $F: \close{B} \times [0,1] \to \rem$ by
  $F(x,t) = (1-t)f_e(x) + tf(x)$.  By the bounds on $\relconst$ and
  $\circrad{\splxs}$, for every $x \in \close{B}$, and $p \in \splxs$,
  we have $\distEm{x}{p} \leq \frac{2\relconst}{\thickbnd\sparseconst}
  + \circrad{\splxs} < 2\samconst$.  Thus it follows from
  \Lemref{lem:close.to.ctr} that $x \in \bdry{\close{B}}$ cannot be a
  $\relconst$-centre.

  It is convenient to consider the max norm on $\rem$ defined by the
  largest magnitude of the components: $\infnorm{f_e(x)} = \max_{i}
  \abs{f_e(x)_i}$. (This gives us a better bound than working with the
  standard Euclidean norm.)  If $\infnorm{f_e(y)} \leq
  \frac{\relconst}{2}$, then $y$ must be a $\relconst$-centre. Indeed,
  we would have $\abs{ \norm{p_i - y} - \norm{p_j - y} } \leq \abs{
    \norm{p_i - y} - \norm{p_0 - y} } + \abs{ \norm{p_0 - y} -
    \norm{p_j - y} } \leq \frac{\relconst}{2} + \frac{\relconst}{2} =
  \relconst$ for all $i$ and $j$. Thus, since $x \in \bdry{\close{B}}$
  is not a $\relconst$-centre, we must have $\infnorm{f_e(x)} >
  \frac{\relconst}{2}$.

  Also, from the hypothesis on $\gdist$, we have
  $\infnorm{f_e(x)-f(x)} \leq 2\pertconst = \frac{\relconst}{2}$, for
  all $x \in \bdry{\close{B}}$.  It follows that $\infnorm{F(x,t)}
  \geq \infnorm{f_e(x)} - t\infnorm{f(x) - f_e(x)} > 0$ when $x \in
  \bdry{\close{B}}$.
\end{proof}

We will need the following observation:
\begin{lem}
  \label{lem:fe.regular}
  The origin is a regular value for the function $f_e$ defined in
  \Eqnref{eq:euclidean.f}.
\end{lem}
\begin{proof}
  Choose a coordinate system such that $\circcentre{\splxs} \in B$ is
  the origin. We show by a direct calculation that $\det J(f_e)_0 \neq
  0$, where $J(f_e)_0$ is the Jacobian matrix representing the
  derivative of $f_e$ in our coordinate system.

  Let $p_{i} = \transp{(p_{i0}, \, \dots, \, p_{im})}$ for all $p_{i}
  \in \{ p_{0}, \, \dots, \, p_{m}\}$.  For $x = \transp{(x_{1}, \,
    \dots, \, x_{m})} \in \reel^{m}$, let $f_e(x) = \transp{(f_{0}(x),
    \,\dots,\, f_{m}(x))}$, where
  \begin{equation*}
    f_{i}(x) = \|p_{i} - x\| - \| p_{0} - x\| =
    \sqrt{\sum_{k=1}^{m} (p_{ik} - x_{k})^{2} }  -
    \sqrt{\sum_{k=1}^{m} (p_{0k} - x_{k})^{2} }\, . 
  \end{equation*}

  We find
  \begin{equation*} 
    \frac{\partial f_{i}}{\partial x_{j}} \Big|_0 =  \frac{p_{0j} -
      p_{ij}}{\circrad{\sigma}}, 
  \end{equation*}
  and thus
  \begin{equation}
    \label{eq:jacobian.f}
    J(f_e)_0 = -\frac{1}{\circrad{\splxs}} \transp{P},
  \end{equation}
  where as usual $P$ is the matrix whose columns are $p_i-p_0$. Since
  $\vol(\sigma^{m}) = \frac{|\det(P)|}{m!}$, \Eqnref{eq:jacobian.f} implies
  \begin{equation*}
    \left| \det  J(f_e)_0 \right| = \frac{m!\,
      \vol(\splxs^{m})}{\circrad{\splxs}^{m}}. 
  \end{equation*}
  Thus since $f_e^{-1}(0) = \{0\}$, $0$ is a regular value for $f_e$
  provided $\splxs$ is non-degenerate.
\end{proof}

\Lemref{lem:metric.cc.pert} now follows from a consideration of the
degree of the mappings $f$ and $f_e$ relative to zero. The
\defn{degree} of a smooth map $g: \close{B} \to \rem$ at a regular
point $p \in g(B)$ is defined by
\begin{equation*}
  \deg (g,p,B) = \sum_{x \in g^{-1}(p)} \sgn \det J(g)_x \, ,
\end{equation*}
where $J(g)_x$ is the Jacobian matrix of $g$ at $x$.
The exposition by Dancer~\cite{dancer2000} is a good reference for the
degree of maps from manifolds with boundary. It is shown that the
definition of $\deg(g,p,B)$ extends to continuous maps $g$ that are
not necessarily differentiable. If $h: \close{B} \to \rem$ is
homotopic to $g$ by a homotopy $H: \close{B} \times [0,1] \to \rem$
such that $H(x,t) \neq p$ for all $t \in [0,1]$, and $x \in \bdry{B}$,
then $\deg (g,p,B) = \deg(h,p,B)$.

Since $f_e^{-1}(0) = \{ \circcentre{\splxs} \}$, it follows from
\Lemref{lem:fe.regular} that $\deg(f_e,0,B) = \pm 1$. Then
\Lemref{lem:fef.homotopic} implies $\deg(f,0,B) = \deg(f_e,0,B)$,
and since this is nonzero, it must be that $f^{-1}(0) \neq \emptyset$.
The demonstration of \Lemref{lem:metric.cc.pert} is complete.

\subsubsection{Circumcentres and point perturbations}

The exact same argument as was used to demonstrate
\Lemref{lem:metric.cc.pert} can be used to show that an $m$-simplex
$\tsplxs = \simplex{\tilde{p}_0, \ldots, \tilde{p}_m}$ whose vertices
lie close to a thick $m$-simplex $\splxs$, will also have a
circumcentre that lies close to $\circcentre{\splxs}$. We replace the
function $f$ defined in \Eqnref{eq:metric.f} by the function
\begin{equation*}
  \tilde{f}(x) =
  \transp{(\distEm{x}{\tilde{p}_1} -
    \distEm{x}{\tilde{p}_0}, \ldots, \distEm{x}{\tilde{p}_m} -
    \distEm{x}{\tilde{p}_0})}, 
\end{equation*}
and the same argument goes through. We obtain:
\begin{lem}[Circumcentres: point perturbation]
  \label{lem:point.cc.pert}
  Suppose that $\splxs = \simplex{p_0,\ldots,p_m}$ is an
  $\thickbnd$-thick $m$-simplex with $\circrad{\splxs} < \samconst$
  and $\shortedge{\splxs} \geq \sparseconst\samconst$. Suppose also
  that $\tilde{\splxs} = \simplex{\tilde{p}_0,\ldots, \tilde{p}_m}$ is
  such that $\norm{\tilde{p}_i - p_i} \leq \pertconst$ for all $i \in
  [0,\ldots,m]$. If
  \begin{equation*}
   \pertconst \leq \frac{\thickbnd \sparseconst \samconst}{8},
   \qquad
   \text{then}
   \qquad
   \distEm{\circcentre{\tsplxs}}{\circcentre{\splxs}} < 
    \frac{8\pertconst}{\thickbnd\sparseconst}.
  \end{equation*}
\end{lem}

\subsection{Perturbations and protection}
\label{sec:pert.prot.ball}

Suppose $\pert: \pts \to \tpts$ is a $\pertconst$-perturbation.  If
$\splxs$ is a $\delta$-protected $m$-simplex in $\delP$, then we want
an upper bound on $\pertconst$ that will ensure that $\tsplxs =
\pert(\splxs)$ is protected in $\delof{\tpts}$.
The following definition will be convenient:
\begin{de}[Secure simplex]
  \label{def:secure}
  A simplex $\splxs \in \delP$ is \defn{secure} if it is a
  $\delta$-protected $m$-simplex that is $\thickbnd$-thick and
  satisfies $\circrad{\splxs} < \samconst$ and $\shortedge{\splxs} \geq
  \sparseconst \samconst$.
\end{de}
Our stability results apply to subcomplexes of secure simplices, the
definition of which employs multiple parameters. For safe interior
simplices \Lemref{lem:delta.sparse} and \Thmref{thm:prot.thick} allow
us to consolidate some of these parameters with the ratio
$\delta/\samconst$:
\begin{lem}[Safe interior simplices are secure]
  \label{lem:safe.int.secure}
  If $\pts$ satisfies a sampling radius $\samconst$ and is
  $\delta$-generic for $\sdipts$, with $\delta = \protconst\samconst$,
  then the safe interior $m$-simplices are secure, with $\sparseconst =
  \protconst$, and $\thickbnd = \frac{\sqrt{3}\protconst^2}{4m}$.
\end{lem}

\begin{lem}[Protection and point perturbation]
  \label{lem:point.pert.prot.ball}
  Sup\-pose that $\pts \subset \rem$ and $\splxs \in \delP$ is secure.
  If $\pert: \pts \to \tpts$ is a $\pertconst$-perturbation with
  \begin{equation*}
    \pertconst \leq \frac{\thickbnd\sparseconst}{18}\delta,
  \end{equation*}
  then $\pert(\splxs) = \tsplxs \in \delof{\tpts}$ and has a $(\delta
  - \frac{18}{\thickbnd\sparseconst} \pertconst)$-protected
  Delaunay ball.
\end{lem}
\begin{proof}
  Let $B = \ballEm{c}{r}$ be the $\delta$-protected Delaunay ball for
  $\splxs \in \delP$, and let $\tilde{B} =
  \ballEm{\tilde{c}}{\tilde{r}}$ be the circumball for the
  corresponding perturbed simplex $\tsplxs$. We wish to establish a
  bound on $\pertconst$ that will ensure that $\tilde{B}$ is protected
  with respect to $\tilde{\pts}$.

  Let $q \in \pts$ be a point not in $\splxs$.  We need to ensure that
  the corresponding $\tilde{q}$ lies outside the closure of
  $\tilde{B}$, i.e., that $\distEm{\tilde{q}}{\tilde{c}} > \tilde{r}$.

  Since $\delta \leq \samconst$, the hypothesis of
  \Lemref{lem:point.cc.pert} is satisfied by $\pertconst$, and we have
  $\distEm{\tilde{c}}{c} < \localconst\pertconst$, where $\localconst
  = \frac{8}{\thickbnd \sparseconst}$.  Thus for $p \in
  \splxs$ and corresponding $\tilde{p} \in \tsplxs$ we have
  \begin{equation*}
    \begin{split}
      \tilde{r} &\leq \distEm{c}{p} +
        \distEm{c}{\tilde{c}} +
        \distEm{p}{\tilde{p}}\\
      &< r + (\localconst + 1)\pertconst.
    \end{split}
  \end{equation*}
  Also
  \begin{equation*}
    \begin{split}
      \distEm{\tilde{q}}{\tilde{c}} &\geq \distEm{q}{c} -
      \distEm{\tilde{c}}{c} - \distEm{\tilde{q}}{q} \\
      &> r +\delta - \pertconst(\localconst + 1).
    \end{split}
  \end{equation*}
  Therefore $\tilde{q}$ will be outside of the closure of $\tilde{B}$
  provided $r + \delta - \pertconst(\localconst + 1) \geq r + (1 +
  \localconst)\pertconst$, i.e., when
$    \delta \geq 2(\localconst + 1)\pertconst$.
  The result follows from the definition of $\localconst$ and the
  observation that $\sparseconst$ and $\thickbnd$ are each no larger
  than one.
\end{proof}

A similar argument yields a bound on the
metric perturbation that will ensure the Delaunay balls for the
$m$-simplices remain protected:
\begin{lem}[Protection and metric perturbation]
  \label{lem:metric.pert.prot.ball}
  Suppose $U \subset \rem$ contains $\convhull{\pts}$ and $\gdist: U
  \times U \to \reel$ is a metric 
  such that $\abs{ \distEm{x}{y} - \dist{x}{y} } \leq \pertconst$ for
  all $x,y \in U$. Suppose also that $\splxs \in \delP$ is secure.
  If
  \begin{equation*}
    \pertconst \leq \frac{\thickbnd\sparseconst}{20}\delta,
  \end{equation*}
  and $\distEm{p}{\bdry{U}} \geq 2\samconst$ for every vertex $p \in
  \splxs$, then $\splxs$ also belongs to $\delPd$, and has a $(\delta
  - \frac{20}{\thickbnd\sparseconst} \pertconst)$-protected
  Delaunay ball in the metric $\gdist$.
\end{lem}
\begin{proof}
  Let $B = \ballEm{c}{r}$ be the Euclidean $\delta$-protected Delaunay
  ball for $\splxs \in \delP$, and let $\tilde{B} =
  \ballEm{\tilde{c}}{\tilde{r}}$ be a circumball for $\splxs$ in the
  metric $\gdist$. We wish to establish a bound on $\pertconst$ that
  will ensure that $\tilde{B}$ is protected with respect to $\gdist$.

  Let $q \in \pts$ be a point not in $\splxs$.  We need to ensure that
  $\dist{q}{\tilde{c}} > \tilde{r}$.  Since $\delta \leq \samconst$,
  the hypothesis ensures that $\pertconst \leq
  \frac{\thickbnd\sparsity}{8}$, and so
  \Lemref{lem:metric.cc.pert} yields $\distEm{\tilde{c}}{c} <
  \localconst\pertconst$, where $\localconst =
  \frac{8}{\thickbnd\sparseconst}$. Thus for $p \in \splxs$
 \begin{equation*}
    \begin{split}
      \tilde{r} &\leq \dist{c}{p} + \dist{c}{\tilde{c}}\\
      &< (r + \pertconst) + (\localconst\pertconst + \pertconst)\\
      &= r + (\localconst + 2)\pertconst,
    \end{split}
  \end{equation*}
  and
 \begin{equation*}
    \begin{split}
      \dist{q}{\tilde{c}} &\geq \dist{q}{c} - \dist{\tilde{c}}{c}\\
     &> r +\delta - (\localconst + 2) \pertconst.
    \end{split}
  \end{equation*}
  Thus $\tilde{q}$ will be outside of the closure of $\tilde{B}$ provided
  $r + \delta - (\localconst + 2)\pertconst \geq r +
  (\localconst + 2)\pertconst$, i.e., when
  \begin{equation*}
    \delta \geq 2(\localconst + 2)\pertconst.
  \end{equation*}
  The result follows from the definition of $\localconst$ and the
  observation that $\sparseconst$ and $\thickbnd$ are each no larger
  than one.
\end{proof}

\subsection{Perturbations and Delaunay stability}
\label{sec:del.stab}

The results of \Secref{sec:pert.prot.ball} translate into stability
results for Delaunay triangulations. In the case of point
perturbations in Euclidean space,  the connectivity of the
Delaunay triangulation cannot change as long as the simplices
corresponding to the initial $m$-simplices remain protected. This is a
direct consequence of Delaunay's original result~\cite{delaunay1934},
but we explicitly lay out the argument.

In the case of metric perturbation, we can no longer take for granted
that the Delaunay complex cannot change its connectivity if the
$m$-simplices remain protected. This is because we are no longer
guaranteed that the Delaunay complex will be a triangulation. Using
the consequences of the point-perturbation result, we establish bounds
that ensure that the Delaunay complex in the perturbed metric will be
the same as the original Delaunay triangulation.

\subsubsection{Point perturbations}

A consequence of Delaunay's triangulation result is that if a
perturbation does not destroy any $m$-simplices in the Delaunay
complex of a generic point set, then no new simplices are created
either, and the complex is unchanged. More precisely we have:
\begin{lem}
  \label{lem:eucl.no.new}
  Suppose $\pts \subset \rem$ is a generic sample set, and $\qpts
  \subseteq \pts$ is a subset of interior points. If $\pert:
  \pts \to \tpts$ is a perturbation such that $\pert(\str{\qpts; \delP})
  \subseteq \str{\pert(\qpts) ; \delof{\tpts}}$, and every $m$-simplex
  $\tsplxs^m \in \pert(\str{\qpts})$ is protected in
  $\delof{\tpts}$, then
$    \pert(\str{\qpts}) = \str{\pert(\qpts)}$.    
\end{lem}
\begin{proof}
  Let $p \in Q$ By \Lemref{lem:embed.del.star}, $\str{\pert(p)}$ is
  embedded, and by \Lemref{lem:Delaunay.local.tri}, $\delP$ is a
  triangulation at $p$. Since $\pert: \pts \to \tpts$ is injective, it
  follows that the simplical map induced by $\pert$ must be injective,
  and the result follows from \Lemref{lem:inject.triang}.
\end{proof}

\Lemref{lem:point.pert.prot.ball} establishes bounds on a
$\pertconst$-perturbation $\pert: \pts \to \tpts$ which will guarantee
that if $\qpts \subset \pts$, and the simplices in $\str{\qpts}$ are
secure, then $\pert(\str{\qpts}) \subseteq \delof{\tilde{P}}$.
\Lemref{lem:point.pert.prot.ball} also guarantees that, if
$\pertconst$ is small enough, the $m$-simplices in $\str{\pert(\qpts);
  \delof{\tpts}}$ will be protected.  Thus if $\qpts$ consists only of
interior points of $\pts$, \Lemref{lem:eucl.no.new} applies.
We have the following stability theorem for protected Delaunay
triangulations:
\begin{thm}[Stability under point perturbation]
 \label{thm:thick.eucl.stability}
 Suppose $\pts \subset \rem$ and $\qpts \subseteq \pts$ is a subset of
 interior points such that every $m$-simplex in $\str{\qpts}$ is
 secure.  If $\pert: \pts \to \tpts$ is a $\pertconst$-perturbation,
 with
  \begin{equation*}
    \pertconst \leq \frac{\thickbnd\sparseconst}{18}\delta
  \end{equation*}
  then
  \begin{equation*}
    \str{\qpts;\delP} \pertiso \str{\pert(\qpts);\delof{\tpts}}.
  \end{equation*}
\end{thm}

The $\pertconst$-relaxed Delaunay complex for $\pts$ was defined by de
Silva~\cite{deSilva2008} by the criterion that $\splxs \in
\relDel{\pertconst}{\pts}$ if and only if there is a ball
$B=\ballEm{c}{r}$ such that $\splxs \subset \close{B}$, and $\distEm{c}{q}
\geq r - \pertconst$ for all $q \in \pts$. Thus the simplices in
$\relDel{\pertconst}{\pts}$ all have ``almost empty'', balls centred
on a $\pertconst$-centre for $\splxs$. 
We have the following consequence of \Thmref{thm:thick.eucl.stability}:
\begin{cor}[Stability under relaxation]
  \label{cor:thick.eucl.stability.relax}
  Suppose $\pts \subset \rem$ and $\qpts \subseteq \pts$ is a set of
  interior points such that every $m$-simplex in $\str{\qpts}$ is
  secure.  If
  \begin{equation*}
    \pertconst \leq \frac{\thickbnd\sparseconst}{18}\delta,
  \end{equation*}
  then
  \begin{equation*}
   \str{\qpts;\relDel{\pertconst}{\pts}} =
    \str{\qpts;\delP}.
  \end{equation*}
\end{cor}
\begin{proof}
  Suppose that $\splxs \in \str{\qpts ; \relDel{\pertconst}{\pts}
  }$.  Then there is a ball $B$ enclosing $\splxs$ such that any point
  $q \in B$ is within a distance $\pertconst$ from $\bdry{B}$. Project
  all such points radially out to $\bdry{B}$. Then we have a
  $\pertconst$-perturbation $\pert: \pts \to \tilde{\pts}$, and
  $\splxs$ has become $\tilde{\splxs} \in \str{\pert(\qpts) ;
    \delof{\tpts}}$.  By \Thmref{thm:thick.eucl.stability},
  $\str{\pert(\qpts) ; \delof{\tpts}} \pertiso \str{\qpts;\delP}$,
  and therefore $\splxs \in \str{\qpts;\delP}$.
\end{proof}

\subsubsection{Metric perturbation}

For a perturbation of the metric, we can exploit the stability results
obtained for perturbations of points in the Euclidean metric to ensure
that no simplices can appear in $\str{\qpts;\delPd}$ that do not already
exist in $\str{\qpts;\delP}$.
\begin{lem}
  \label{lem:thick.metric.no.new}
  Suppose $\convhull{\pts} \subseteq U \subset \rem$ and $\gdist: U
  \times U \to \reel$ is such that $\abs{\dist{x}{y} - \distEm{x}{y}}
  \leq \pertconst$ for all $x,y \in U$. Suppose also that $\qpts
  \subseteq \pts$ is a set of interior points such that every
  $m$-simplex $\splxs \in \str{\qpts}$ is secure and satisfies
  $\distEm{p}{\bdry{U}} \geq 2\samconst$ for every vertex $p \in
  \splxs$.  If
  \begin{equation*}
    \pertconst \leq \frac{\thickbnd\sparseconst}{36}\delta,    
  \end{equation*}
  then
  \begin{equation*}
    \str{\qpts;\delPd} \subseteq \str{\qpts;\delP}.
  \end{equation*}
\end{lem}
\begin{proof}
  Let $\ball{c}{r}$ be a Delaunay ball for simplex $\splxs \in
  \str{\qpts;\delPd}$. Then $\dist{c}{p} \leq \dist{c}{q}$ for all $p
  \in \splxs$, and $q \in \pts$. By the hypothesis on $\gdist$, this
  implies that $\distEm{c}{p} \leq \distEm{c}{q} + 2\pertconst$ for
  all $p \in \splxs$ and $q \in \pts$, and therefore $\splxs \in
  \relDel{2\pertconst}{\pts}$. The result now follows from
  \Corref{cor:thick.eucl.stability.relax}.
\end{proof}

The perturbation bounds required by \Lemref{lem:thick.metric.no.new}, also
satisfy the requirements of \Lemref{lem:metric.pert.prot.ball}.
This gives us the reverse inclusion, and thus we can quantify the
stability under metric perturbation for subcomplexes of secure
simplices in Delaunay triangulations:
\begin{thm}[Stability under metric perturbation]
  \label{thm:thick.metric.stability}
  Suppose $\convhull{\pts} \subseteq U \subset \rem$ and $\gdist: U
  \times U \to \reel$ is such that $\abs{\dist{x}{y} - \distEm{x}{y}}
  \leq \pertconst$ for all $x,y \in U$. Suppose also that $\qpts
  \subseteq \pts$ is a set of interior points such that every
  $m$-simplex $\splxs \in \str{\qpts}$ is secure and satisfies
  $\distEm{p}{\bdry{U}} \geq 2\samconst$ for every vertex $p \in
  \splxs$. If
  \begin{equation*}
    \pertconst \leq \frac{\thickbnd\sparseconst}{36}\delta,
  \end{equation*}
  then
  \begin{equation*}
    \str{\qpts;\delPd} = \str{\qpts;\delP}.
  \end{equation*}
\end{thm}
Using \Lemref{lem:safe.int.secure}, and recognizing that the safe
interior simplices also satisfy the distance from boundary requirement
of \Thmref{thm:thick.metric.stability}, we can restate this metric
perturbation stability result for Delaunay triangulations on
$\delta$-generic point sets:
\begin{cor}[Stability under metric perturbation]
  \label{cor:metric.stability}
  Suppose $\pts$ is $\delta$-generic for $\sdipts$, with sampling
  radius $\samconst$ and $\delta = \protconst\samconst$. Suppose also
  that $\convhull{\pts} \subseteq U$, and $\gdist: U \times U \to \reel$
  is such that $\abs{\dist{x}{y} - \distEm{x}{y}} \leq \pertconst$ for
  all $x,y \in U$. If
  \begin{equation*}
    \pertconst \leq \frac{\protconst^3}{84m}\delta
    = \frac{\protconst^4}{84m}\samconst,
  \end{equation*}
  then
  \begin{equation*}
    \str{\sdipts;\delPd} = \str{\sdipts;\delP}.
  \end{equation*}
\end{cor}

%

\section{Conclusions}

We have quantified the close relationship between the genericity of a
point set, the quality of the simplices in the Delaunay complex, and
the stability of the Delaunay complex under perturbation.  The problem
of poorly shaped simplices in a higher dimensional Delaunay complex
can be seen as a manifestation of point sets that are close to being
degenerate.  The introduction of thickness as a geometric quality
measure for simplices facilitated the stability calculations, which
develop around a consideration of the circumcentres of a simplex in
the presence of a perturbation.

We considered a point set $\pts \subset \rem$ meeting a sampling
radius $\samconst$ and showed a constant bound on the thickness of the
Delaunay simplices provided $\pts$ is $\delta$-generic with $\delta =
\protconst \samconst$ for some constant $\protconst$. The question then
arises: What is the least upper bound on the feasible $\protconst$
as a function of the dimension $m$?

In a companion paper \cite{boissonnat2013flatpert.inria}, we develop a perturbation algorithm which produces a $\delta$-generic point set from a given $\samconst$-sample set.  Since the triangulation results of the current work are localised, we can extend the perturbation algorithm to construct Delaunay triangulations of abstract Riemannian manifolds that are not necessarily embedded in an ambient space, as we have shown in subsequent work~\cite{boissonnat2013manmesh.inria}.  The idea is that a manifold can be locally well approximated by Euclidean space, so we fit together local Euclidean Delaunay patches where the Euclidean metric varies slightly between patches. This is where the stability of the Delaunay patches is important. In this setting we can also accommodate variations in the sampling radius between neighbouring patches. Thus the algorithm is able to triangulate sample sets whose sampling radius is defined by a Lipschitz sizing function.





\subsection*{Acknowledgements}

The authors would like to thank David Cohen-Steiner for suggesting the
proof technique employed in \Lemref{lem:metric.cc.pert}. Comments and
suggestions from the reviewers considerably improved the final
document; the authors gratefully acknowledge this contribution.

This work was partially supported by the CG Learning project. The
project CG Learning acknowledges the financial support of the Future
and Emerging Technologies (FET) programme within the Seventh Framework
Programme for Research of the European Commission, under FET-Open
grant number: 255827.

\phantomsection
\bibliographystyle{alpha}
\addcontentsline{toc}{section}{Bibliography}
\bibliography{delrefs}

\newcommand{\etalchar}[1]{$^{#1}$}
\begin{thebibliography}{AGG{\etalchar{+}}10}

\bibitem[AGG{\etalchar{+}}10]{agarwal2010scg}
P.~K. Agarwal, J.~Gao, L.~Guibas, H.~Kaplan, V.~Koltun, N.~Rubin, and
  M.~Sharir.
\newblock Kinetic stable {D}elaunay graphs.
\newblock In {\em SoCG}, pages 127--136, 2010.

\bibitem[BDG13a]{boissonnat2013flatpert.inria}
J.-D. Boissonnat, R.~Dyer, and A.~Ghosh.
\newblock {D}elaunay stability via perturbations.
\newblock Research Report RR-8275, INRIA, 2013.

\bibitem[BDG13b]{boissonnat2013manmesh.inria}
J.-D. Boissonnat, R.~Dyer, and A.~Ghosh.
\newblock {D}elaunay triangulation of manifolds.
\newblock Research Report RR-8389, INRIA, 2013.

\bibitem[BG11]{boissonnat2011tancplx}
J.-D. Boissonnat and A.~Ghosh.
\newblock Manifold reconstruction using tangential {D}elaunay complexes.
\newblock Technical Report RR-7142 v.3, INRIA, 2011.
\newblock (To appear in DCG).

\bibitem[BS04]{bandyopadhyay2004}
D.~Bandyopadhyay and J.~Snoeyink.
\newblock Almost-{D}elaunay simplices: nearest neighbor relations for imprecise
  points.
\newblock In {\em SODA}, pages 410--419, 2004.

\bibitem[BWY11]{boissonnat2011aniso.tr}
J.-D. Boissonnat, C.~Wormser, and M.~Yvinec.
\newblock Anisotropic {D}elaunay mesh generation.
\newblock Technical Report RR-7712, INRIA, 2011.
\newblock (To appear in SIAM J. of Computing).

\bibitem[CDE{\etalchar{+}}00]{cheng2000}
S.-W. Cheng, T.~K. Dey, H.~Edelsbrunner, M.~A. Facello, and S.~H Teng.
\newblock Sliver exudation.
\newblock {\em Journal of the {ACM}}, 47(5):883--904, 2000.

\bibitem[CDR05]{cheng2005}
S.-W. Cheng, T.~K. Dey, and E.~A. Ramos.
\newblock Manifold reconstruction from point samples.
\newblock In {\em SODA}, pages 1018--1027, 2005.

\bibitem[CG12]{canas2012}
G.~D. Ca{\~n}as and S.~J. Gortler.
\newblock Duals of orphan-free anisotropic {V}oronoi diagrams are embedded
  meshes.
\newblock In {\em SoCG}, pages 219--228, New York, NY, USA, 2012. ACM.

\bibitem[Dan00]{dancer2000}
E.~N. Dancer.
\newblock Degree theory on convex sets and applications to bifurcation.
\newblock In {\em Calculus of Variations and Partial Differential Equations},
  pages 185--225. Springer-Verlag, 2000.

\bibitem[Del34]{delaunay1934}
B.~Delaunay.
\newblock Sur la sph\`ere vide.
\newblock {\em Izv. Akad. Nauk SSSR, Otdelenie Matematicheskii i Estestvennyka
  Nauk}, 7:793--800, 1934.

\bibitem[dS08]{deSilva2008}
V.~de~Silva.
\newblock A weak characterisation of the {D}elaunay triangulation.
\newblock {\em Geometriae Dedicata}, 135:39--64, 2008.

\bibitem[Dug66]{dugundji1966}
J.~Dugundji.
\newblock {\em Topology}.
\newblock Allyn and Bacon, Inc., Boston, 1966.

\bibitem[DZM08]{dyer2008sgp}
R.~Dyer, H.~Zhang, and T.~M\"{o}ller.
\newblock Surface sampling and the intrinsic {V}oronoi diagram.
\newblock {\em Computer Graphics Forum (Special Issue of Symp. Geometry
  Processing)}, 27(5):1393--1402, 2008.

\bibitem[ES97]{edelsbrunner1997rdt}
H.~Edelsbrunner and N.~R. Shah.
\newblock Triangulating topological spaces.
\newblock {\em Int. J. Comput. Geometry Appl.}, 7(4):365--378, 1997.

\bibitem[FKMS05]{funke2005}
S.~Funke, C.~Klein, K.~Mehlhorn, and S.~Schmitt.
\newblock Controlled perturbation for {D}elaunay triangulations.
\newblock In {\em SODA}, pages 1047--1056, 2005.

\bibitem[Li03]{li2003}
X-Y. Li.
\newblock Generating well-shaped $d$-dimensional {D}elaunay meshes.
\newblock {\em Theoretical Computer Science}, 296(1):145--165, 2003.

\bibitem[LL00]{leibon2000}
G.~Leibon and D.~Letscher.
\newblock {D}elaunay triangulations and {V}oronoi diagrams for {R}iemannian
  manifolds.
\newblock In {\em SoCG}, pages 341--349, 2000.

\bibitem[LS03]{labelle2003}
F.~Labelle and J.~R. Shewchuk.
\newblock Anisotropic {V}oronoi diagrams and guaranteed-quality anisotropic
  mesh generation.
\newblock In {\em SoCG}, pages 191--200, 2003.

\bibitem[Mun68]{munkres1968}
J.~R. Munkres.
\newblock {\em Elementary differential topology}.
\newblock Princton University press, second edition, 1968.

\bibitem[Mun84]{munkres1984}
J.~R. Munkres.
\newblock {\em Elements of Algebraic Topology}.
\newblock Addison-Wesley, 1984.

\bibitem[TB97]{trefethen1997}
L.N. Trefethen and D.~Bau.
\newblock {\em Numerical linear algebra}.
\newblock Society for Industrial Mathematics, 1997.

\bibitem[Whi57]{whitney1957}
H.~Whitney.
\newblock {\em Geometric Integration Theory}.
\newblock Princeton University Press, 1957.

\end{thebibliography}

\end{document}